\documentclass[journal,letterpaper,onecolumn,11pt]{IEEEtran}
\usepackage{hyperref}
\usepackage{tikz}
\usetikzlibrary{positioning, arrows.meta, fit, backgrounds, calc}
\usetikzlibrary{arrows.meta}
\usepackage{amsmath}
\hypersetup{
 colorlinks,
 linkcolor={blue!100!black},
 citecolor={blue!100!black},
 urlcolor={blue!80!black}
}
\usepackage{accents}
\newlength{\dhatheight}

\def\BibTeX{{\rm B\kern-.05em{\sc i\kern-.025em b}\kern-.08em
    T\kern-.1667em\lower.7ex\hbox{E}\kern-.125emX}}
\usepackage{stmaryrd} 

\newcommand{\bF}{\mathbb{F}}

\usepackage{booktabs} 
\newcommand{\cA}{\mathcal{A}}

\newcommand{\cC}{\mathcal{C}}
\newcommand{\cD}{\mathcal{D}}
\newcommand{\cE}{\mathcal{E}}
\newcommand{\cF}{\mathcal{F}}

\newcommand{\cI}{\mathcal{I}}

\newcommand{\cK}{\mathcal{K}}

\newcommand{\cS}{\mathcal{S}}
\newcommand{\cT}{\mathcal{T}}

\newcommand{\boldc}{\mathbf{c}}

\newcommand{\boldf}{\mathbf{f}}

\newcommand{\boldp}{\mathbf{p}}
\newcommand{\boldq}{\mathbf{q}}
\newcommand{\boldr}{\mathbf{r}}
\newcommand{\bolds}{\mathbf{s}}

\newcommand{\boldv}{\mathbf{v}}
\newcommand{\boldw}{\mathbf{w}}
\newcommand{\boldx}{\mathbf{x}}
\newcommand{\boldy}{\mathbf{y}}

\newcommand{\rank}{\operatorname{rank}}

\usepackage{tikz}
\newif\ifFULL
\FULLfalse
\usepackage[utf8]{inputenc} 
\usepackage[T1]{fontenc}
\usepackage{url}
\usepackage{bbm}
\usepackage{ifthen}
\usepackage[cmex10]{mathtools}

\usepackage{verbatim}
\usepackage{amsthm,amssymb}
\usepackage{xcolor}
\usepackage{graphicx}
\usepackage{amssymb}
\usepackage{epstopdf}

\newtheorem{theorem}{Theorem}
\newtheorem{observation}{Observation}
\newtheorem{remark}{Remark}

\newtheorem{example}{Example}
\newtheorem{lemma}{Lemma}
\newtheorem{definition}{Definition}

\newtheorem{corollary}{Corollary}
\usepackage[style=ieee,citestyle=numeric-comp,sorting=nty]{biblatex}
\usepackage{algorithm}
\usepackage{algorithmic}
\newcommand{\blue}[1]{\textcolor{blue}{#1}}
\newcommand{\gray}[1]{\textcolor{gray}{#1}}

\begin{document}
\title{Improved Torn Paper Coding via Local Alignment} 


\author{%
  \IEEEauthorblockN{\textbf{Junsheng Liu} and \textbf{Netanel Raviv}}  \IEEEauthorblockA{\\~\\Department of Computer Science and Engineering\\
                    Washington University in St Louis, St Louis, MO, USA\\
                    \texttt{junsheng,netanel.raviv@wustl.edu}}
}


\maketitle


\begin{abstract}
   In the torn paper channel, a transmitted codeword is broken at random locations into fragments that arrive at the decoder in an unordered manner.
   A central theoretical challenge within this model is global alignment---the task of determining each fragment's original position---in order to faithfully reconstruct the entire codeword. 
   Prior work by Shomorony and Vahid introduced an interleaved-pilot scheme that successfully achieved a vanishing error probability. 
   However, their alignment strategy relies heavily on global statistics, requiring fragments to exceed a minimum length and effectively discarding many shorter ones as erasures, which results in rates significantly below capacity.
   To address this gap, we propose an improved coding scheme that achieves a provable rate increase through a novel approach we call \textit{local alignment}.   
   This approach identifies global alignment bits within each fragment using only local information, allowing the decoder to determine the positions of fragments that are shorter than those used in previous work. 
   Consequently, the decoder can extract information from a much larger fraction of the channel output than in previous work, yielding significantly higher rates.
  Furthermore, we extend our analysis to torn paper coding with lost pieces (TPC-LP), a generalized model that accounts for length-dependent fragment deletion. For a class of TPC-LP channels that delete all fragments below a logarithmic length threshold while allowing arbitrary length-dependent deletion probabilities for longer fragments, we show that the proposed local alignment strategy achieves an arbitrarily small additive gap to capacity as the threshold increases.
\end{abstract}

\section{Introduction}
{\let\thefootnote\relax\footnote{
Parts of this work accepted for publication in the 2026 IEEE International Symposium on Information Theory. This work was supported in part by NSF grant CNS-2223032.}}

Consider a simple game: you write an information sequence on a long strip of paper, but before the receiver can read it, a friend plays a trick by tearing the paper into random, shuffled pieces. 
While the text on each individual fragment remains perfectly legible, all context regarding their original order is completely lost. 
To win the game and recover the message, you must find a way to piece the fragments back together, based only on the information written on each fragment.
This intuitive scenario captures the core mechanics of torn paper coding, a fundamental reconstruction problem in information and coding theory.

This reconstruction problem arises naturally in emerging data storage and transmission systems that are vulnerable to physical fragmentation, with DNA storage and forensic 3D fingerprinting providing key examples.
DNA storage has emerged as an important alternative for large-scale archival systems~\cite{douglas2012logic,goldman2013towards,blawat2016forward,antkowiak2020low}, but its biochemical nature makes DNA molecules susceptible to random strand breaks. This vulnerability motivates the torn-paper coding problem studied herein, where information must be reconstructed from unordered fragments.
Beyond synthetic biology, this fragmentation problem also appears in physical embedding of data within 3D-printed objects. 
To enable secure forensic identification, fingerprinting information is encoded directly into the physical structure of a 3D-printed model~\cite{wang2024secureinformationembeddingextraction}.
However, if the object is broken---whether accidentally or through malicious tampering---the embedded data string is shattered. 
An adversary might even conceal certain pieces, leaving the decoder to recover the information from an incomplete, unordered set of fragments.

These physical vulnerabilities—from degrading DNA to shattered 3D prints—lead to a natural coding-theoretic abstraction: an encoder transmits a codeword through a channel that breaks it into pieces, shuffles them, and potentially discards some. The objective is to reliably recover the encoded information from the surviving, unordered fragments.

Torn paper coding was initially proposed in the seminal work by Shomorony and Vahid~\cite{shomorony2021torn}, which established the channel capacity and proposed initial encoding strategies. 
This framework was later generalized by Ravi, Vahid and Shomorony~\cite{ravi2024recovering} to account for lost pieces, i.e., where each fragment is independently discarded with a length-dependent probability. 
In parallel, practical designs in~\cite{nassirpour2023dna, jiao2025efficient} utilize nested codes to achieve high empirical rates; however, these heuristic approaches currently lack rigorous guarantees that the error probability vanishes as the codeword length approaches infinity.

Beyond torn paper coding, several adjacent models address related fragmentation challenges under different constraints.
Motivated specifically by 3D printing forensics, break-resilient codes~\cite{wang2024breakresilientcodesforensic3d}  are designed to tolerate an adversary introducing up to a fixed number of breaks, ensuring that every resulting fragment carries enough positional information to allow full recovery.
Similarly, \cite{liu2025single} studied a model in which the decoder must reconstruct the data from only a single, sufficiently large fragment.
The adversarial torn paper model of~\cite{barlev} approaches the problem by restricting the lengths of possible fragments, while the sliced channel~\cite{sima2023error} examines a setting where the sequence is partitioned into fixed-length slices that may suffer internal substitution or deletion errors.
While these adversarial and structured models provide valuable insights, in this paper we focus on torn paper coding, where fragments are generated by random cuts.

Under these random cuts, the fundamental challenge is global alignment, i.e., recovering the original sequence from an unordered multiset of fragments. 
To tackle this, prior work introduced an interleaved-pilot scheme~\cite{shomorony2021torn} that embeds a De Bruijn  pilot sequence periodically among the symbols of a randomly shifted erasure code. 
Since sufficiently long fragments contain a sufficiently long portion of the pilot sequence,
the decoder is able to determine the absolute global position of such fragments. 
However, to reliably distinguish the pilot subsequence from other data within the fragment, the decoder must rely on global statistical uniqueness, which imposes a relatively large restriction on minimum fragment length.
Consequently, the decoder is forced to establish a minimum fragment length threshold, discarding any fragment that falls below it. 
By treating a substantial volume of these shorter fragments as erasures, this approach leaves a significant gap in the achievable transmission rate.

To close this gap, we propose an improved coding scheme that shifts the alignment paradigm from global statistical uniqueness to local structural patterns.
Specifically, we impose a Run-Length Limited (RLL) block constraint on the erasure-coded information to strictly cap the maximum number of consecutive zeros.
We then strategically embed a sequence violation---a deterministic marker of consecutive zeros exceeding this cap---exclusively within the De Bruijn pilot sequence. 
Since this embedded marker is mathematically prohibited from appearing anywhere in the erasure-coded information, the decoder can distinguish the pilot-sequence bits from the erasure-coded bits within any fragment that contains the marker. 
This local alignment method effectively \textit{halves} the minimum required fragment length. 
By safely extracting positional information from much shorter fragments without the need for long statistical uniqueness, our scheme utilizes a significantly larger fraction of the channel output, yielding strictly higher achievable rates compared to existing provable schemes, while maintaining a vanishing decoding error probability.
The improved capacity in shown in Figure~\ref{Figure: comparison for rate of different methods}, with comparison between the rates of~\cite{shomorony2021torn}, the rates of~\cite{Junsheng2026torn}, the channel capacity~\cite{shomorony2021torn}, and the rates in this work.
\begin{figure}[t]
    \centering
    \includegraphics[width=0.5\linewidth]{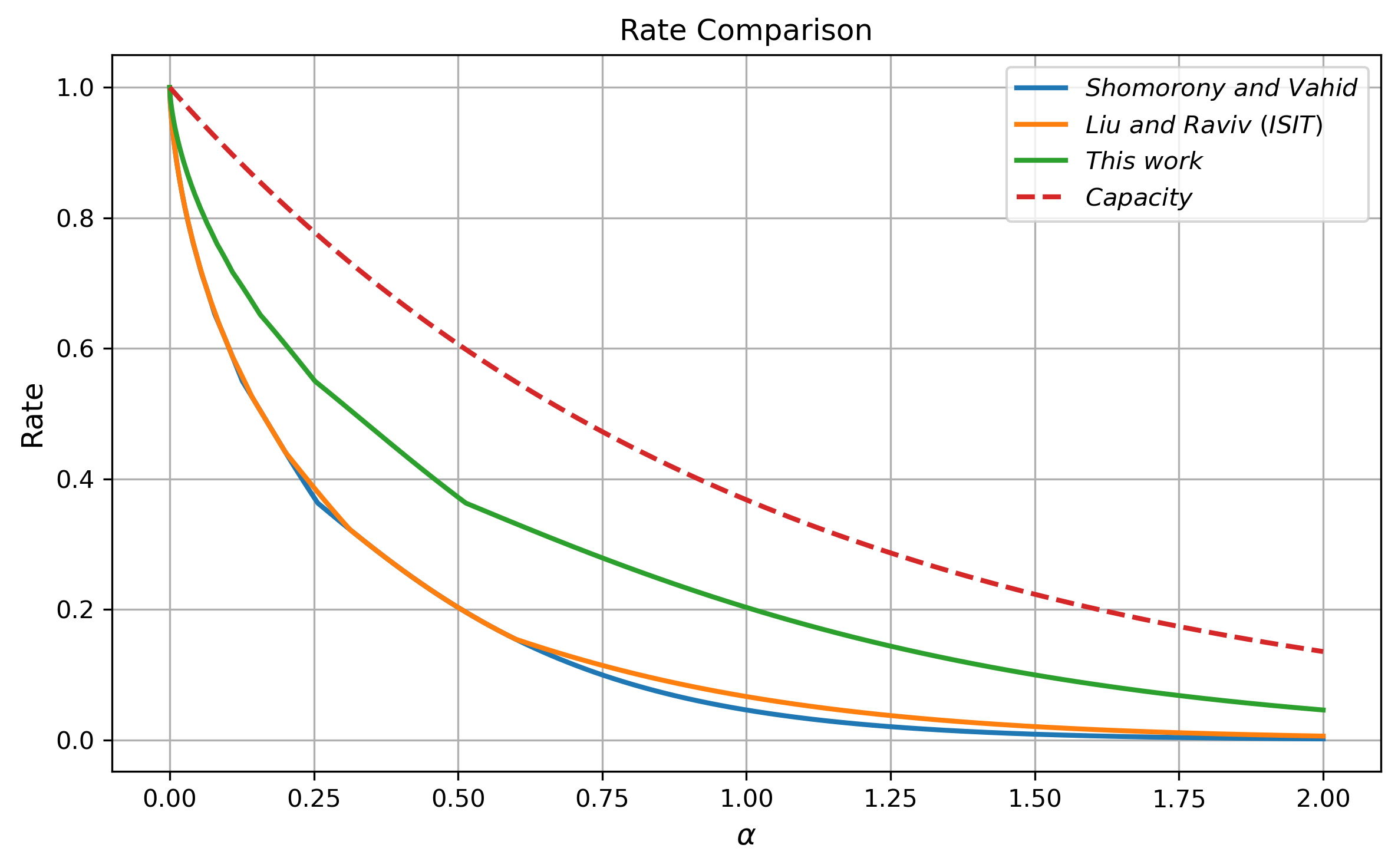}
    \caption{ A comparison between the rates of~\cite{shomorony2021torn}, the rates of~\cite{Junsheng2026torn}, the channel capacity~\cite{shomorony2021torn}, and the rates in this work. }
    \label{Figure: comparison for rate of different methods}
\end{figure}

Furthermore, the advantages of this local alignment strategy extend naturally to torn paper coding with lost pieces (TPC-LP)~\cite{ravi2024recovering}. 
In this generalized model, fragments are independently discarded (or ``deleted'') with a length-dependent probability, mirroring physical scenarios where shorter pieces are entirely lost. 
Since our proposed encoding treats short fragments as erasures, it is inherently robust to such deletions.
We analytically demonstrate that our method not only guarantees reliable recovery under such fragment deletions, but for a specific parameter regime it achieves an arbitrarily small gap to the theoretical capacity of the TPC-LP.

The remainder of this paper is organized as follows. Section~\ref{Section:problem definition} formally defines the torn paper coding model, reviews the capacity result of~\cite{shomorony2021torn}, and discusses the interleaved-pilot scheme of~\cite{shomorony2021torn} which serves as the starting point for our construction. 
Section~\ref{Section:propsed work} presents the proposed local-alignment framework. 
In particular, Section~\ref{section:rll} introduces the run-length-limited constructions, Section~\ref{section: improved encoding} describes the encoding procedure, Section~\ref{Section:improved decoding} develops the decoding algorithm, and Section~\ref{section:calculation of rate} rigorously proves reliable recovery (i.e., vanishing decoding error probability as the codeword length goes to infinity) and derives the achievable rate of the proposed scheme. Section~\ref{section:comparison} numerically compares the optimized achievable rates of the proposed method with prior provable constructions and with the torn-paper channel capacity. 
Finally, Section~\ref{section:compare capacity for tpclp} extends the local-alignment framework to torn paper coding with lost pieces, where we show that under a threshold-deletion regime the resulting scheme achieves an arbitrarily small additive gap to the TPC-LP capacity.
\section{Problem Definition and Previous Results}\label{Section:problem definition}
We consider the problem of encoding 
a message in $\{0,1\}^k$ into a codeword~$\boldx=(x_1,\ldots,x_n)\in\{0,1\}^n$. 
During storage or transmission, the codeword is subject to random physical breaks between consecutive bits. 
For each index $i\in \{1,2,\dots,n-1\}$, a break occurs between $x_i$ and $x_{i+1}$ with probability $p_n$ and no break with probability $1-p_n$, where~$p_n$ is a channel parameter which may depend on~$n$.
The fragmentation process does not delete or alter any bits; it only introduces random cuts between consecutive positions. 
More specifically, let $N_1, N_2, \ldots$ be i.i.d. $\operatorname{Geometric}(p_n)$ random variables representing the lengths of the broken pieces, and let $K$ be the smallest index such that $$\sum_{i=1}^{K} N_i \ge n.$$
The channel tears the string $\boldx$ into fragments $\vec{X}_1, \vec{X}_2, \ldots, \vec{X}_K$, where 
$$\vec{X}_i \triangleq \left( x_{1+\sum_{j=1}^{i-1} N_j}, \dots, x_{\sum_{j=1}^i N_j} \right)$$
for~$1\le i<K$, and the final fragment is bounded by the total length $n$:$$\vec{X}_K \triangleq \left( x_{1+\sum_{j=1}^{K-1} N_j}, \dots, x_n \right).$$

The decoder receives the unordered multiset $\{\{\vec{X}_1, \vec{X}_2, \dots, \vec{X}_K\}\}$ of all resulting fragments, and the goal is to design a family of binary codewords $\mathcal{C} \subseteq \{0,1\}^n$ such that the probability of decoding error—defined as $P_e = \max_{c \in \mathcal{C}} \text{Pr}(\psi(\{\{\vec{X}_j\}\}_{j=1}^K) \neq c)$ for a decoding function $\psi$ that maps the multiset of fragments back to the original codeword—vanishes as $n \to \infty$.
Clearly, for any practical use the function~$\psi$ must be efficiently computable. 
The channel is illustrated in Figure~\ref{Figure: channel}, where raw information is encoded into a string, the string is torn into pieces, and the decoder must decode the correct information based on the unordered set of substrings.

We assume that the received fragments retain their orientation, i.e., when receiving a fragment~$(x_i,\ldots,x_{i+\ell})$, the decoder does not confuse this fragment with its reverse $(x_{i+\ell},\ldots,x_{i})$.
This assumption reflects DNA storage settings, where the underlying biochemical directionality of the DNA strand dictates the correct reading direction~\cite{wiki:Directionality_molecular_biology} and is similarly mirrored in 3D-printing forensics where asymmetric inner codes and orientation-specific markers are embedded into the model’s geometry to break rotational symmetry and ensure unique readability~\cite{wang2024secureinformationembeddingextraction}.
Thus, the orientation of each fragment is assumed to be known.

\begin{figure}[ht]
    \centering
    \includegraphics[width=0.6\linewidth]{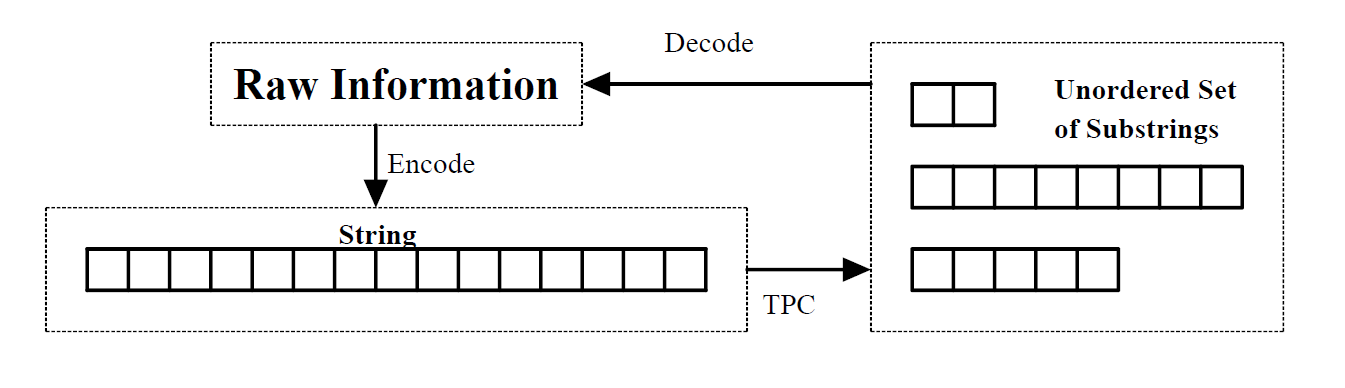}
    \caption{A sketch of torn paper coding.}
    \label{Figure: channel}
\end{figure}

The capacity of the associated channel (called the \textit{torn paper channel)} has been studied in~\cite{shomorony2021torn}.
They defined $\alpha=\lim_{n\to\infty}p_n\log n$ and proved that the capacity of this channel is $e^{-\alpha}$.
In the case that $p_n=o(\frac{1}{\log n})$, we have that $e^{-\alpha}$ approaches~$1$ when $n\xrightarrow[]{}\infty$, and similarly, if $p_n=\omega(\frac{1}{\log n})$ then $e^{-\alpha}$ approaches~$0$ when $n\xrightarrow[]{}\infty$.
Thus, the remaining regime of interest is $p_n=\Theta(\frac{1}{\log n})$, i.e., where $\alpha$ is constant.

Ref.~\cite{shomorony2021torn} also shows that purely random codes achieve this capacity.
As purely random codes are impractical for encoding and decoding, the authors introduce an interleaving framework based on combining a De Bruijn sequence with a shifted pseudorandom erasure code. 
In this framework, the De Bruijn sequence serves as a global pilot structure for fragment localization; however, this requires fragments to be long enough to successfully distinguish pilot markers from information bits. We provide technical details in Section~\ref{section:interleaved pilot scheme}, and extend this idea using a technique we call \textit{local alignment} in Section~\ref{section: improved encoding}.
This technique uses local structures within a fragment to identify the locations of De Bruijn sequence bits, and thereby relaxes the fragment length requirement and leads to higher rates.

Ref.~\cite{ravi2024recovering} studied torn paper coding with lost pieces (TPC-LP), which extends torn paper coding to accommodate lost fragments by introducing a length-dependent deletion probability. 
In this model, an encoded binary sequence first passes through the standard torn paper channel. 
Then, each resulting fragment $\vec{X}_i \in \{\{\vec{X}_1, \vec{X}_2, \ldots, \vec{X}_K\}\}$ is independently discarded with some probability~$d(\cdot)$ that depends on fragment length. 
The channel output is the unordered multiset of the surviving fragments, once again assuming no internal bit-level errors (the full details of TPC-LP are given in Section~\ref{section:compare capacity for tpclp}). 
As demonstrated in Section~\ref{section:compare capacity for tpclp}, our proposed coding method achieves an arbitrarily small gap to the capacity of the TPC-LP in a specific fragment deletion regime.

A complementary line of work studies practical torn paper coding schemes through empirical evaluation, but does not provide rigorous analytical guarantees, i.e., does not establish that the decoding error probability vanishes as $n \to \infty$.
For instance,~\cite{nassirpour2023dna,jiao2025efficient} propose nested designs that achieve higher rates than~\cite{shomorony2021torn} in practice.
Thus, while these approaches offer useful algorithmic constructions, their fundamental limits and provable performance remain unknown.

\subsection{Interleaved-pilot Scheme}\label{section:interleaved pilot scheme}
Since our method is similar to the interleaved pilot scheme of~\cite{shomorony2021torn}, we introduce it here as the basis for our refined construction in Section~\ref{section: improved encoding}.
We begin with the definition of a De Bruijn sequence. 
\begin{definition}\label{Definition:De Bruijn sequence}
A \emph{De Bruijn sequence} of order $\log n$ over $\mathbb{F}_2$ is a binary sequence of length $n$ in which every possible binary string of length $\log n$ appears exactly once as a (possibly cyclic) contiguous substring.
\end{definition}

For a fixed positive integer $m$, consider a De Bruijn sequence $\mathbf{p}$ of length $n/m$, which we refer to as the \emph{pilot} sequence. 
According to Definition~\ref{Definition:De Bruijn sequence}, every sequence of length $\log(n/m)$ appears in $\mathbf{p}$ exactly once.
To construct the codebook, the authors of~\cite{shomorony2021torn} interleave the pilot sequence with codewords of an erasure code. 
Specifically, let $\cC_{\mathrm{er}}$ be an erasure code of length $n/m$ and rate $R_{\mathrm{er}}$. They form a \emph{shifted} version of this code by drawing an i.i.d.\ $\mathrm{Ber}(1/2)$ random vector in $\{0,1\}^{n/m}$ and XORing it with each codeword of $\cC_{\mathrm{er}}$, thereby obtaining the pseudo-random code $\Tilde{\cC}_{\mathrm{er}}$.
A key property required for successful global alignment is that the pilot sequence substrings should not coincide with any substring of a randomly shifted erasure codeword as follows.
\begin{observation}\label{Lemma: probability of match in ilan's method}
    For any $\mathbf{s} \in \Tilde{\cC}_{\mathrm{er}}$, the probability that $\mathbf{s}$ shares a length-$k$ substring with the pilot sequence is
\begin{align*}
    \Pr\Big(&\mathbf{p}[i:i+k-1]=\mathbf{s}[j:j+k-1] \text{ for any distinct }i,j\in[\tfrac{n}{m}-k]
    \Big) \leq (\tfrac{n}{m})^22^{-k},
\end{align*} 
where the probability is taken over the choices of the random vector used to construct~$\tilde{\cC}_\text{er}$ from~$\cC_\text{er}$, and where $\mathbf{p}[i:i+k-1]$ refers to the $i$-th through $(i+k-1)$-th bits of~$\mathbf{p}$.
\end{observation}

Observation~\ref{Lemma: probability of match in ilan's method} is proved via a straightforward application of the union bound. 
Therefore, if $k=(2+\delta)\log n$ for some $\delta>0$, this probability goes to $0$ when $n\xrightarrow[]{}\infty$. 
This means that for any $\epsilon>0$ there exists a large enough~$n$ so that it is possible to choose at least a $1-\epsilon$ fraction of the codewords in $\Tilde{\cC}_{\mathrm{er}}$ such that \textit{any} $(2+\delta)\log n$ segment in \textit{any} of its codewords is distinct from \textit{any} $(2+\delta)\log n$ segment of the pilot sequence~$\mathbf{p}$.

Let $\cS=\{\bolds_1,\bolds_2,\dots,\bolds_{|\cS|}\}\subset\Tilde{\cC}_{\mathrm{er}}$ be a set of $(1-\epsilon)2^{\frac{n}{m}R_{\mathrm{er}}}$ such codewords, i.e., where none of them shares any $(2+\delta)\log n$ contiguous subsequence with $\mathbf{p}$. 
Ref.~\cite{shomorony2021torn} builds a codeword $\boldc$ by taking $m-1$ codewords 
from $\cS$ and interleaving their symbols with symbols from $\mathbf{p}$. 
More precisely, for any function~$u:[m-1]\to[|\cS|]$, they build the codeword $\boldc_u=(\boldc_u[0],\boldc_u[1],\ldots, \boldc_u[n-1])$ where for every $t\in\{0,1,\dots,n/m-1\}$, we have
\begin{equation}
    \boldc_u[mt+j]=
    \begin{cases}
    \mathbf{p}[t], & j= 0, \\
    \bolds_{u(j)}[t], & j=1, 2,3,\dots,m-1.
\end{cases}
\end{equation}
The resulting codebook is of size $|\cS|^{m-1}=(1-\epsilon)^{m-1}2^{(1-\frac{1}{m})nR_{\mathrm{er}}}$. 

Subsequently, decoding begins by recovering the global alignment of the fragments through the periodically embedded pilot symbols drawn from the De Bruijn sequence $\mathbf{p}$. 
Since a window of length $(2+\delta)\log n$ from the pilot sequence coincides with a window of the shifted erasure code with probability which tends to zero as $n \to \infty$, every fragment whose length exceeds $(2+\delta)m\log n$ must contain at least one pilot window of length $(2+\delta)\log n$ that is uniquely identifiable. 
Moreover, because every $\log(n/m)$ window of the De Bruijn sequence $\mathbf{p}$ is distinct, the occurrence of such a window immediately reveals the fragment’s absolute position in the codeword, a process we call \textit{global alignment}. 
Once these long fragments are anchored, the short fragments—those too small to contain any identifiable pilot window—are discarded, and the resulting missing positions are treated as erasures. 
The reconstruction of the full message is then completed by the outer shifted-erasure code $\Tilde{\cC}_{\mathrm{er}}$, which deterministically corrects the erased locations created by the random breaking process.

Notice that in the interleaved-pilot scheme, the decoder identifies the locations of the pilot sequence bits only in sufficiently long fragments. 
The rationale is that the probability of a long window of the De Bruijn sequence coinciding with a window of any shifted erasure codeword is vanishingly small. 
However, this approach has an inherent gap: the number of bits required to ensure that a De Bruijn window is \emph{not} identical to a shifted erasure codeword window is approximately  $2\log n$ (Observation~\ref{Lemma: probability of match in ilan's method}), whereas the number of bits needed to uniquely identify a De Bruijn window is only $\log (n/m)$. 
Consequently, the decoder in~\cite{shomorony2021torn} restricts itself to fragments of length at least $(2+\delta)m\log n$, discarding the large number of shorter fragments, which reduces the effective rate since all such fragments are treated as erasures.

This observation motivates the question of whether one can identify the locations of the pilot sequence bits by exploiting the local pattern surrounding each De Bruijn bit, rather than relying solely on long fragments. 
Specifically, we embed a small deterministic marker into the De Bruijn sequence and actively prevent this marker from appearing within the erasure code. 
By detecting these distinct structural patterns within any given fragment, we can distinguish the De Bruijn bits from the erasure-coded bits—a technique we term \textit{local alignment}. 
Once this pattern is identified, the De Bruijn window can be reliably inferred. 
This enables the decoder to determine the global position of fragments whose lengths are well below $(2+\delta)m\log n$.

\section{Improved Torn Paper Coding via Local Alignment}\label{Section:propsed work}
To implement this local alignment, our proposed scheme constrains both the erasure code $\cC_{\mathrm{er}}$ and the De Bruijn sequence using a Run-Length Limited (RLL) $(0, k)$ block code. We then strategically embed a sequence violations---specifically, multiple substrings of $k+1$ consecutive $0$'s called \textit{markers}---exclusively within the De Bruijn sequence to serve as  local alignment anchors. 
Since this violation is strictly prohibited from occurring within the RLL-constrained $\cC_{\mathrm{er}}$-codeword (denoted~$\bar{\cC}_{\mathrm{er}}$), as well as from the remainder of the De Bruijn sequence, the marker\blue{s} remains entirely distinguishable. 
Consequently, the modified De Bruijn sequence functions as a dual-purpose pilot: it acts as a delimiter to identify local boundaries within a fragment while simultaneously establishing global alignment. 
This dual functionality ensures alignment integrity, facilitating the recovery of short segments that would otherwise be discarded, ultimately achieving a measurable increase in the overall information rate.

\subsection{Run Length Limited Code}\label{section:rll}
Since our method uses an RLL technique for local alignment, we introduce it here as the basis for our construction in Section~\ref{section: improved encoding}.
The codewords of an RLL code are binary sequences of a fixed length $n$, characterized by parameters $d$ and $k$. 
These parameters indicate the minimum and maximum number of $0$'s between consecutive $1$'s in the sequence. 
For example, an $\text{RLL}(1,3)$ code, also known as Modified Frequency Modulation (MFM)~\cite{immink1990runlength}, requires that between any two $1$'s there are at least $d=1$ and at most $k=3$ $0$'s.
To meet the requirements of Section~\ref{section: improved encoding} which follows, 
we require an~$\mathrm{RLL}(0,k-1)$ scheme and an~$\mathrm{RLL}(0,k-2)$ scheme, which are obtained by insertion of~$1$'s into the information words, and are introduced in the following simple lemmas. 


\begin{lemma}\label{lemma: there exists an RLL fixed length code}
    There exists an $\mathrm{RLL}(0,k-1)$ fixed-length block code with rate $\frac{k-1}{k}$.
\end{lemma}
\begin{proof}
    Consider a binary input word $\boldv\in\{0,1\}^l$, and assume for simplicity that $k-1$ divides $l$ (otherwise, pad the last block with at most $k-2$ dummy bits, which does not affect the asymptotic rate).
    Partition $\boldv$ into $\frac{l}{k-1}$ consecutive blocks of length $k-1$, say $\boldv^{(1)},\boldv^{(2)},\dots,\boldv^{(l/(k-1))}$, and append a~$1$ to every block:
    \[
        \boldx^{(i)} \triangleq \boldv^{(i)}1,\qquad i=1,\dots,\frac{l}{k-1}.
    \]
    The final codeword is the concatenation $\boldx \triangleq \boldx^{(1)}\boldx^{(2)}\cdots \boldx^{(l/(k-1))}$, and therefore has length $\frac{lk}{k-1}$, yielding rate $\frac{l}{lk/(k-1)}=\frac{k-1}{k}$.
    Since every length-$k$ block ends with a $1$, any $k$ consecutive bits contain a~$1$ and hence any zero run has length at most $k-1$; in particular, $\boldx$ satisfies the $\mathrm{RLL}(0,k-1)$ constraint.
    Decoding simply reverses this construction: split the received word into length-$k$ blocks and delete the last bit of each block.
\end{proof}
While Lemma~\ref{lemma: there exists an RLL fixed length code} established the existence of an $\mathrm{RLL}(0, k-1)$ fixed-length block code with rate $\frac{k-1}{k}$, the following lemma introduces an $\mathrm{RLL}(0, k-2)$ fixed-length block code achieving a rate of $\frac{k-2}{k}$.
\begin{lemma}\label{lemma: there exists another RLL fixed length code}
    There exists an $\mathrm{RLL}(0,k-2)$ fixed-length block code with rate $\frac{k-2}{k}$.
\end{lemma}
\begin{proof}
    Consider a binary input word $\boldv\in\{0,1\}^l$, and assume for simplicity that $k-2$ divides $l$ (otherwise, pad the last block with at most $k-3$ dummy bits, which does not affect the asymptotic rate).
    Partition $\boldv$ into $\frac{l}{k-2}$ consecutive blocks of length $k-2$, say $\boldv^{(1)},\boldv^{(2)},\dots,\boldv^{(l/(k-2))}$, and prepend and append a~$1$ to each block:
    \[
        \boldx^{(i)} \triangleq 1\boldv^{(i)}1,\qquad i=1,\dots,\frac{l}{k-2}.
    \]
    The final codeword is the concatenation $\boldx \triangleq \boldx^{(1)}\boldx^{(2)}\cdots \boldx^{(l/(k-2))}$, and therefore has length $\frac{lk}{k-2}$, yielding rate $\frac{l}{lk/(k-2)}=\frac{k-2}{k}$.
    Moreover, since every length-$k$ block begins and ends with a $1$, any $k$ consecutive bits contain two $1$'s and hence any zero run has length at most $k-2$; in particular, $\boldx$ satisfies the $\mathrm{RLL}(0,k-2)$ constraint.
    Decoding simply reverses this construction: split the received word into length-$k$ blocks and delete the first bit and the last bit of each block to recover~$\boldv$.
\end{proof}
Lemmas~\ref{lemma: there exists an RLL fixed length code} and \ref{lemma: there exists another RLL fixed length code} demonstrate two distinct methods for constructing fixed-length block codes that achieve an asymptotic rate of $1$ as $k \to \infty$. 
Furthermore, these two constructions are tailored to serve different purposes within the overall encoding process in Section~\ref{section: improved encoding}.
To simplify our exposition and explicitly distinguish between these two schemes, we adopt the following terminology:
\begin{itemize}
\item We refer to the $\mathrm{RLL}(0,k-1)$ fixed-length block code with rate $\frac{k-1}{k}$ from Lemma~\ref{lemma: there exists an RLL fixed length code} as the \textit{$\mathrm{RLL}(0,k-1)$ code}, and its corresponding operations as the \textit{$\mathrm{RLL}(0,k-1)$ encoder and decoder}.
\item We refer to the $\mathrm{RLL}(0,k-2)$ fixed-length block code with rate $\frac{k-2}{k}$ from Lemma~\ref{lemma: there exists another RLL fixed length code} as the \textit{$\widetilde{\mathrm{RLL}}(0,k-2)$ code}, and its corresponding operations as the \textit{$\widetilde{\mathrm{RLL}}(0,k-2)$ encoder and decoder}.
\end{itemize}
\subsection{Encoding}\label{section: improved encoding}

The encoding process starts with an inner code designed to align fragments. 
Once fragments are aligned, their missing lengths and positions are known, making this fundamentally an erasure channel (see proofs in Section~\ref{section:calculation of rate}).
By employing a random binary linear code, we can recover these erasures efficiently in polynomial time—essentially by using Gaussian elimination to solve a system of linear equations over $\mathbb{F}_2$. 
The parameters for this code are defined below.
\begin{definition}\label{definition:random linear erasure code}
Fix sufficiently small constants $\delta,\eta>0$ and a constant positive integer~$m$ (construction parameters that will be optimized later), and let $\gamma\triangleq(1+\delta)m$.
Further, let $n_{\mathrm{er}} \triangleq \frac{(\beta-1)n}{m\beta}$ for~$\beta=\beta(n)$ such that~$\beta=o(\log(n))$ and~$\beta=\omega(1)$.
We let $\cC_{\mathrm{er}}$ be a \emph{random binary linear code} with a parity-check matrix $H\in\mathbb{F}_2^{r\times n_{\mathrm{er}}}$ whose entries are chosen independently and uniformly from $\mathbb{F}_2$,
where $r=\left( 1-(\alpha\gamma+1)e^{-\alpha\gamma}+3\eta\right)n_{\mathrm{er}}$ and where $\alpha$ is the channel parameter as defined in the beginning of Section~\ref{Section:problem definition}.
\end{definition}

The following lemma proves that the matrix $H$ achieves full row rank with high probability, thereby establishing the rate of $\cC_{\mathrm{er}}$.

\begin{lemma}\label{lemma:full rank random parity matrix}
Let $H\in\mathbb{F}_2^{r\times n_{\mathrm{er}}}$ be a random binary matrix
whose entries are chosen independently and uniformly from $\mathbb{F}_2$,
where $r\le n_{\mathrm{er}}$. Then, $\operatorname{rank}(H)=r$ with high probability.
Consequently, the random binary linear code $\cC_{\mathrm{er}}$ defined by the
parity-check matrix $H$ has rate
$R_{\mathrm{er}}=1-\frac{r}{n_{\mathrm{er}}}$.
In particular, if
$r=\left(1-(\alpha\gamma+1)e^{-\alpha\gamma}+3\eta\right)n_{\mathrm{er}}$,
then
$R_{\mathrm{er}}=(\alpha\gamma+1)e^{-\alpha\gamma}-3\eta$
with high probability.
\end{lemma}

\begin{proof}
The matrix $H$ fails to have full row rank only if some row lies in the span of
the previous rows. Conditioned on the first $i-1$ rows being linearly
independent, their span contains $2^{i-1}$ vectors out of the total
$2^{n_{\mathrm{er}}}$ vectors in $\mathbb{F}_2^{n_{\mathrm{er}}}$. Hence the
probability that the $i$-th row lies in this span is
$2^{i-1-n_{\mathrm{er}}}$.
By the union bound,
$\Pr(\operatorname{rank}(H)<r)\le\sum_{i=1}^{r}2^{i-1-n_{\mathrm{er}}}\le2^{r-n_{\mathrm{er}}}$.
Since
$r-n_{\mathrm{er}}=-\left((\alpha\gamma+1)e^{-\alpha\gamma}-3\eta\right)n_{\mathrm{er}}$,
the probability above tends to zero exponentially fast in
$n_{\mathrm{er}}$, provided $\eta$ is sufficiently small. Therefore,
$\operatorname{rank}(H)=r$ with high probability.
Finally, when $\operatorname{rank}(H)=r$, the null space of $H$ has dimension
$n_{\mathrm{er}}-r$, so the corresponding binary linear code has rate
$1-r/n_{\mathrm{er}}$. Substituting the value of $r$ gives
$R_{\mathrm{er}}=(\alpha\gamma+1)e^{-\alpha\gamma}-3\eta$.
\end{proof}
With $\cC_{\mathrm{er}}$ established, the next step in the encoding process is constructing the pilot sequence.
Let $\mathbf{q}$ be a De Bruijn sequence of length $\frac{n}{m}(1-\frac{\beta}{(1+\delta/2)\log n})\frac{\beta-2}{\beta}$, where $\delta$ is a small constant that satisfies $\delta > \frac{2}{\beta-2}$ and $\beta$ is as in Definition~\ref{definition:random linear erasure code}, i.e.~$\beta=o(\log(n))$ and~$\beta=\omega(1)$. 
This sequence $\mathbf{q}$ has the property that every sequence of length $\log(\frac{n(\beta-2)}{m\beta}(1-\frac{\beta}{(1+\delta/2)\log n}))$ appears in $\mathbf{q}$ exactly once.
Subsequently, we apply the $\widetilde{\text{RLL}}(0,\beta-2)$ encoder from Lemma~\ref{lemma: there exists another RLL fixed length code} to~$\mathbf{q}$, yielding a modified De Bruijn sequence $\bar{\mathbf{q}}$ of length $\frac{n}{m}(1-\frac{\beta}{(1+\delta/2)\log n})$. 
For simplicity, we further assume that $\beta|(1+\delta/2)\log n$ and define the modified pilot sequence $\mathbf{p}$ as follows.
\begin{definition}\label{Definition: marker sequence for improved encoding}
    Define the pilot sequence $\mathbf{p}=(\mathbf{p}[0],\mathbf{p}[1],\dots,\mathbf{p}[n/m-1])$ as
\begin{align*}
    \mathbf{p}[t(1+\delta/2)\log n+j]=
    \begin{cases}
    0, & 0 \leq j \leq \beta-1, \\
    \bar{\mathbf{q}}[t(1+\delta/2)\log n + (j-\beta)], & \beta \leq j \leq (1+\delta/2)\log n-1
\end{cases},
\end{align*}
where $t \in \{0,1,\dots, \frac{n}{(1+\delta/2)m\log n}-1\}$ and $j \in \{0, 1, \dots, (1+\delta/2)\log n-1\}$.
\end{definition}

In Definition~\ref{Definition: marker sequence for improved encoding} we systematically embed runs of $\beta$ consecutive zeros, called \textit{markers}, at fixed intervals throughout the pilot sequence $\mathbf{p}$. 
The process of constructing~$\boldp$ from the De Bruijn sequence~$\boldq$ is therefore:
\begin{align}\label{equation:pilot,debruijn and modified de bruijn}
\boldq
\xrightarrow{\quad \widetilde{\text{RLL}}(0,\beta-2) \quad}
\bar{\boldq}
\xrightarrow{\quad \text{embed~$\beta$ $0$-runs} \quad}
\boldp
\end{align}
Based on this process, we establish the following key property of $\mathbf{p}$.
\begin{lemma}\label{lemma:a run of consecutive zeros occurs}
A run of $\beta$ consecutive zeros in~$\boldp$ occurs if and only if it corresponds to an embedded marker.
\end{lemma}
\begin{proof}
The ``if'' direction is immediate: every embedded marker is a $0$-run of length $\beta$, and therefore yields a run of $\beta$ consecutive zeros in $\boldp$.
To prove the converse, suppose for contradiction that there exists a run
\[
    (\boldp[j] , \boldp[j+1] , \ldots , \boldp[j+\beta-1])
    =
    \mathbf{0}_\beta,
\]
where~$\mathbf{0}_\beta$ is a vector with~$\beta$ zero, that does not correspond to an embedded marker for some $j$.
We split into cases according to the position~$j$ from which this run originates.

Recall that $\bar{\boldq}$ is obtained by the
$\widetilde{\mathrm{RLL}}(0,\beta-2)$ encoder by Lemma~\ref{lemma: there exists another RLL fixed length code}, which maps every block
$\boldv\in\{0,1\}^{\beta-2}$ in~$\boldq$ to $(1, \boldv ,1)$.
As $\bar{\boldq}$ is a concatenation of length-$\beta$ blocks and each block begins and ends with a $1$, no length-$\beta$ run of zeros can lie completely inside $\bar{\boldq}$.

It remains to consider a run which intersects a marker but is not exactly that marker. 
Since the marker itself has length $\beta$, such a run must cross one of the two boundaries between a marker and bits that originate in~$\bar{\boldq}$. 
By the assumption that $\beta \mid (1+\delta/2)\log n$, the markers are aligned with the length-$\beta$ block structure of $\bar{\boldq}$. 
Therefore the symbol of $\bar{\boldq}$ immediately before a marker is the final $1$ of some block of the form
$(1,\boldv, 1)$, and the symbol of $\bar{\boldq}$ immediately after a marker is the initial~$1$ of some block of the same form.
Consequently, any length-$\beta$ window that overlaps a marker but is not exactly the marker contains one of these boundary $1$'s, and hence cannot be equal to $\mathbf{0}_\beta$.
\end{proof}
Next, we apply the $\text{RLL}(0,\beta-1)$ fixed-length block encoding algorithm from Lemma~\ref{lemma: there exists an RLL fixed length code} to all codewords of $\cC_{\mathrm{er}}$, yielding a code $\bar{\cC}_{\mathrm{er}}$; the block length of~$\bar{\cC}_{\mathrm{er}}$ is $n/m$ and its rate is $(1-\frac{1}{\beta})R_{\mathrm{er}}$ and hence its size is~$2^{\frac{n}{m}(1-\frac{1}{\beta})R_{\mathrm{er}}}$.
To construct the final code, we interleave the pilot sequence~$\mathbf{p}$ with the codewords of $\bar{\cC}_{\mathrm{er}}$ in a manner that resembles the interleaving from~\cite{shomorony2021torn} given in Section~\ref{section:interleaved pilot scheme}. 
Rather than interleaving~$m-1$ codewords from a subset of a random coset of~$\cC_{\mathrm{er}}$ (which has no substring intersection with their pilot sequence) with a De Bruijn sequence as done in~\cite{shomorony2021torn}, we interleave~$m-1$ codewords of~$\bar{\cC}_{\mathrm{er}}$ with our pilot sequence (Definition~\ref{Definition: marker sequence for improved encoding}).

Formally, we fix an arbitrary indexing~$\bar{\cC}_{\mathrm{er}}=\{\bar{\bolds}_1,\ldots,\bar{\bolds}_{|\bar{\cC}_{\mathrm{er}}|}\}$, and write our code as
\begin{align}\label{eq:maincode}
   \mathcal{C}\triangleq \{\boldc_v=(\boldc_{v}[0], \boldc_{v}[1], \dots, \boldc_{v}[n-1])~\vert~v:[m-1]\to[|\bar{\cC}_{\mathrm{er}}|]\},
\end{align}
where for each $t \in \{0, 1, \dots, n/m-1\}$,
\begin{equation}\label{equation: interleved eq for local alignment}
    \boldc_{v}[mt+j] = 
    \begin{cases} 
    \mathbf{p}[t], & j = 0, \\
    \bar{\bolds}_{v(j)}[t], & j = 1, 2, 3, \dots, m-1
    \end{cases}.
\end{equation}
Since the encoding process above is injective (i.e., $\boldc_{v}\ne\boldc_{v'}$ whenever~$v\ne v'$), the size of the resulting $\mathcal{C}$~\eqref{eq:maincode} is $ |\bar{\cC}_{\mathrm{er}}|^{m-1}=(2^{\frac{n}{m}(1-\frac{1}{\beta})R_{\mathrm{er}}})^{m-1} = 2^{(1-\frac{1}{\beta})(1-\frac{1}{m})nR_{\mathrm{er}}}$.
Since~$\beta=o(\log(n))$ and~$\beta=\omega(1)$, we have $$\lim_{n\to\infty}\left(1-\frac{1}{\beta}\right)=1,$$ and hence the resulting rate tends to~$\left(1-1/m\right)R_{\mathrm{er}}$ as~$n$ tends to infinity, i.e., 
\begin{align}\label{equation:rate}
    \lim_{n\to\infty}\left(1-\frac{1}{\beta}\right)\left(1-\frac{1}{m}\right)R_{\mathrm{er}}=\left(1-\frac{1}{m}\right)R_{\mathrm{er}}.
\end{align}

\subsection{Decoding}\label{Section:improved decoding}
The decoding process essentially reduces the torn paper channel into a (not memoryless) erasure channel by exploiting the interleaved structure of codewords.
At a high level, this is achieved through a three-step procedure:

\begin{enumerate}
    \item[\textbf{A}] \textbf{Local Alignment:} The decoder identifies the positions of the~$\boldp$ bits (see~\eqref{equation:pilot,debruijn and modified de bruijn}) within all fragments of length at least $(1+\delta)m\log n$ and discards all other fragments.
    For each such fragment, this is done by first de-interleaving it to obtain sequences~$\boldw_0,\ldots,\boldw_{m-1}$, and identifying the unique~$\boldw_j$ which contains~$0$-runs of length~$\beta$ as the one which contains bits from~$\boldp$.
    \item[\textbf{B}] \textbf{Global Alignment:} In each sufficiently long fragment, we remove the~$0$-runs of length~$\beta$ from the  $\boldw_j$ found in step \textbf{A}, obtaining a substring from $\bar{\boldq}$, which is then decoded to obtain a substring of~$\boldq$ (again see~\eqref{equation:pilot,debruijn and modified de bruijn}).
    This substring of~$\boldq$---being a De Bruijn subsequence---is used to determine the unique global position of the fragment. 
    \item[\textbf{C}] \textbf{RLL and Erasure Decoding:} After global alignment, we obtain an estimate~$\hat{\boldc}_v$ of $\boldc_v$, and extract substrings of the codewords in $\bar{\cC}_{\mathrm{er}}$ from~$\hat{\boldc}_v$ (see~\eqref{equation: interleved eq for local alignment}). 
    These substrings are processed by an RLL decoding algorithm to remove run-length constraints, followed by standard erasure decoding for $\cC_{\mathrm{er}}$ to recover the original message.
\end{enumerate}

\begin{figure}[ht]
    \centering
    \includegraphics[width=0.6\linewidth]{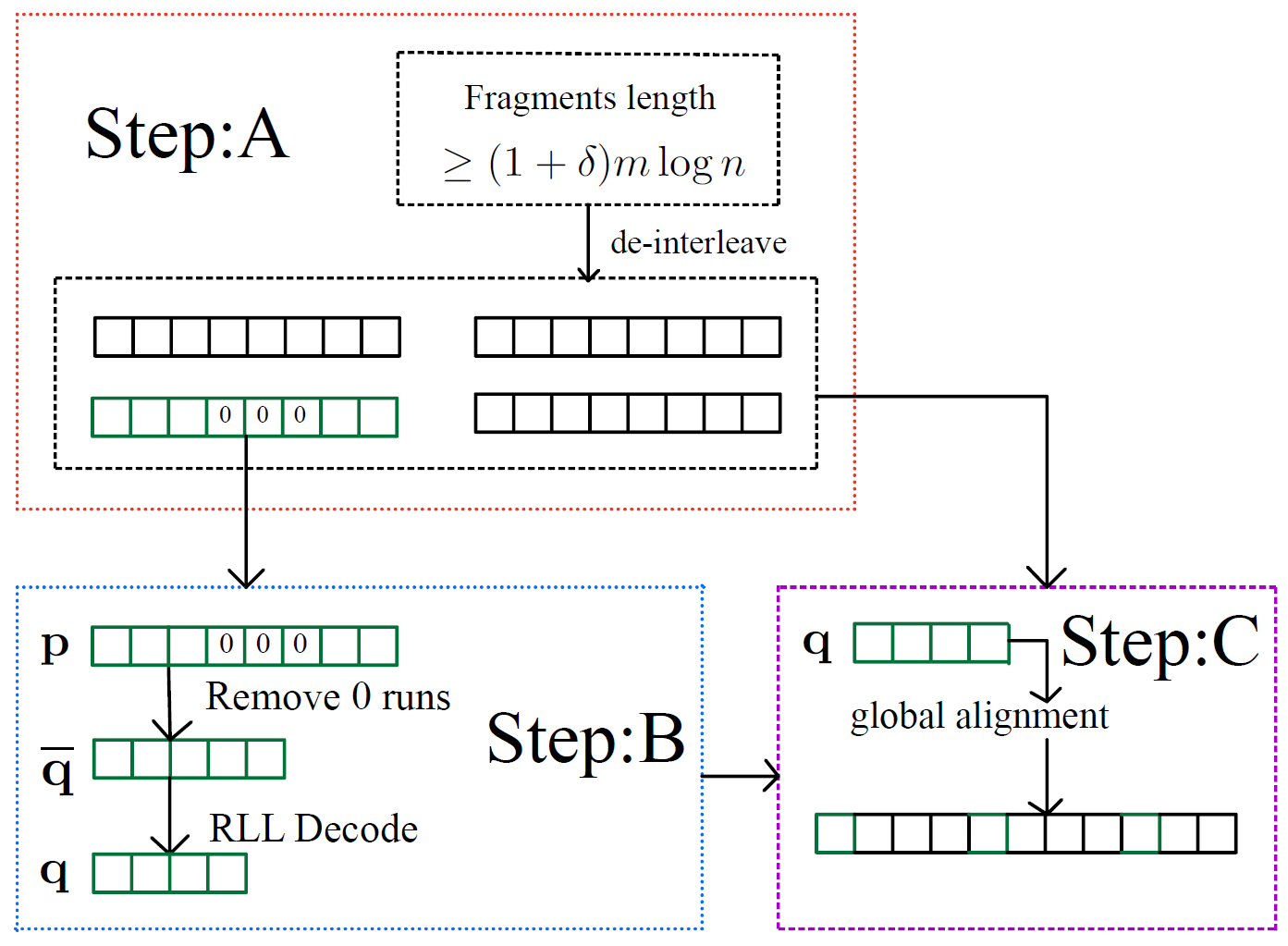}
    \caption{A sketch of the decoding process.}
    \label{Figure: decoding}
\end{figure}

Step \textbf{A} (local alignment) is detailed in Algorithm~\ref{algorithm for finding locations of p in improved decoding}, whose correctness is established in Lemma~\ref{lemma: find pilot bits in rll encoding} below. 
Throughout the proof, recall that a transmitted word~$\boldc_v$ is obtained by interleaving~$\boldp$ with codewords~$\bar{\bolds}_{v(1)},\ldots,\bar{\bolds}_{v(m-1)}$ from~$\bar{\cC}_{\text{er}}$ for some~$v:[m-1]\to |\bar{\cC}_\text{er}|$.
Additionally, let~$\cF$ be the set of all fragments whose length is at least~$(1+\delta)m\log n$.

\begin{algorithm}
\caption{Local alignment.}
\begin{algorithmic}[1]\label{algorithm for finding locations of p in improved decoding}
\STATE \textbf{Input:} Fragment $\boldf=(\boldf[0],\boldf[1],\dots,\boldf[l-1])$ of length $l \geq (1+\delta)m\log n$.
\STATE \textbf{Goal:} Identify the substring $\boldp_{\boldf}$ of the pilot sequence~$\boldp$ in the fragment~$\boldf$, and a substring $\boldy_{\boldf}^i$ of the codeword~$\bar{\bolds}_{v(i)}$ in~$\boldf$, for each~$i\in[m-1]$.
\STATE Initialize empty strings $\boldw_0,\boldw_1,\dots,\boldw_{m-1}$.
\FOR{$i = 0$ to $l-1$} \label{line:deinterleave_start1}
    \STATE Append $\boldf[i]$ to $\boldw_{i \bmod m}$\label{line:deinterleave_start2}
\ENDFOR
\FOR{$j = 0$ to $m-1$}
    \IF{$\boldw_j$ contains a~$0$-run of length~$\beta$}\label{line:identify_start1}
        \STATE $\mathbf{p}_\boldf \leftarrow \boldw_{j \bmod m}$ \label{line:identify_start2}
        \STATE $\mathbf{y}_\boldf^1,\mathbf{y}_\boldf^2,\dots,\mathbf{y}_\boldf^{m-1}\leftarrow \boldw_{(j+1) \bmod m}\dots \boldw_{(j+m-1) \bmod m}$
    \ENDIF
\ENDFOR
\end{algorithmic}
\end{algorithm}
\begin{lemma}\label{lemma: find pilot bits in rll encoding}
    For every~$\boldf\in \cF$, Algorithm~\ref{algorithm for finding locations of p in improved decoding} correctly identifies the bits~$\boldp_{\boldf}$ of $\boldp$, and the bits $\boldy_\boldf^i$ of every $\bar{\bolds}_{v(i)}$, within $\boldf$.
\end{lemma}
\begin{proof}
The decoding of fragment $\boldf\in \cF$ begins by a de-interleaving operation that extracts $m$ candidate substrings, denoted $\boldw_0,\ldots,\boldw_{m-1}$ (Lines~\ref{line:deinterleave_start1}-\ref{line:deinterleave_start2}). 
On the one hand, codewords in $\bar{\cC}_{\mathrm{er}}$ satisfy an $\text{RLL}(0,\beta-1)$ constraint, ensuring that no codeword contains a $0$-run of length~$\beta$. 
On the other hand, each codeword $\boldc_v$ contains a bit from the pilot sequence $\mathbf{p}$ at every $m$-th position, and hence any fragment of $\boldc_v$ of length at least $(1+\delta)m\log n$ contains at least $(1+\delta)\log n$ bits from~$\mathbf{p}$.

By Definition~\ref{Definition: marker sequence for improved encoding} and Lemma~\ref{lemma:a run of consecutive zeros occurs}, it follows that any consecutive subsequence of length at least~$(1+\delta/2)\log n+\beta-1$ of~$\boldp$ contains a~$0$-run of length~$\beta$ which corresponds to an embedded marker.
For sufficiently large $n$, we have $$(1+\delta)\log n>(1+\delta/2)\log n+\beta-1$$ as $\beta=o(\log(n))$.
Therefore, at least one $0$-run of length~$\beta$ is located inside a unique de-interleaved sequence~$\boldw_j$ (Lines~\ref{line:identify_start1}-\ref{line:identify_start2}).
Furthermore, once the identity of the unique~$\boldw_j$ from~$\boldp$ is known, we have that~$\boldw_{(j+1)\bmod m}$ is a substring of~$\bar{\bolds}_{v(1)}$, as well as that $\boldw_{(j+2)\bmod m}$ is a substring of~$\bar{\bolds}_{v(2)}$, and so on.
\end{proof}

Once the bits~$\boldp_{\boldf}$ of the pilot sequence $\boldp$ and the bits $\mathbf{y}_\boldf^i$ of every~$\bar{\bolds}_{v(i)}$ have been located within every fragment, step~\textbf{A} (local alignment) is concluded. 
Then, step~\textbf{B} (global alignment) is conducted via the following sub-steps, which are executed in Algorithm~\ref{algorithm: finding locations of de bruijn}.
\begin{enumerate}
    \item[\textbf{B1}] \textbf{Extraction of \textit{modified} De Bruijn sequence bits:} For every~$\boldf\in \cF$, we identify and extract the substring~$\bar{\boldq}_{\boldf}$ of~$\bar{\boldq}$ inside~$\boldp_{\boldf}$.
    \item[\textbf{B2}] \textbf{Extraction of De Bruijn sequence bits:} For every~$\boldf\in \cF$, we run the $\widetilde{\mathrm{RLL}}(0,\beta-2)$ decoder over~$\bar{\boldq}_{\boldf}$ to extract the bits $\mathbf{q}_{\boldf}$ of~$\boldq$ inside~$\boldp_{\boldf}$.
    \item[\textbf{B3}] \textbf{Perform global alignment:} Determine the global positions of every~$\boldf\in \cF$ using~$\boldq_{\boldf}$.
    
\end{enumerate}

The correctness of steps \textbf{B1} and \textbf{B2} is proved in Lemma~\ref{lemma:correctly identify modified db from pilot} and Lemma~\ref{lemma: modified de brujin sequence can be decode}; the fact that~$\boldq_{\boldf}$ is sufficiently long to conduct \textbf{B3} is proved in Lemma~\ref{lemma:global alignment for qN}.

\begin{lemma}\label{lemma:correctly identify modified db from pilot}
For every~$\boldf\in \cF$, the subsequence $\bar{\mathbf{q}}_{\boldf}$ of the modified De Bruijn sequence~$\bar{\boldq}$ inside~$\boldf$ is uniquely recoverable from the received pilot fragment $\mathbf{p}_{\boldf}$.
\end{lemma}
\begin{proof}
By Lemma~\ref{lemma:a run of consecutive zeros occurs}, any occurrence of $0$-runs of length $\beta$  within the received fragment $\mathbf{p}_{\boldf}$ corresponds uniquely to an inserted alignment marker.
Hence, the subsequence $\bar{\mathbf{q}}_{\boldf}$ is recovered by locating all $0$-runs of length $\beta$ within $\mathbf{p}_{\boldf}$, removing them, and concatenating the remaining segments.
\end{proof}
By Lemma~\ref{lemma:correctly identify modified db from pilot}, the bits corresponding to the modified De Bruijn sequence~$\bar{\mathbf{q}}_{\boldf}$ can be identified and extracted from every~$\boldf\in \cF$.
Following that, we apply an $\widetilde{\mathrm{RLL}}(0,\beta-2)$ decoder over $\bar{\mathbf{q}}_{\boldf}$ to obtain the De Bruijn sequence bits $\mathbf{q}_{\boldf}$, (Line~\ref{line:check if 0 run}), and use~$\boldq_\boldf$ for global alignment.
The detailed process is given in Algorithm~\ref{algorithm: finding locations of de bruijn} and proved in Lemma~\ref{lemma: modified de brujin sequence can be decode}.
\begin{algorithm}
\caption{Extraction of De Bruijn subsequence $\boldq_{\boldf}$ from a pilot fragment $\boldp_{\boldf}$} 
\begin{algorithmic}[1]\label{algorithm: finding locations of de bruijn}
\STATE \textbf{Input:} Pilot fragment $\mathbf{p}_{\boldf} = (\mathbf{p}_{\boldf}[0], \dots, \mathbf{p}_{\boldf}[l-1])$ of length $l \geq (1+\delta)\log n$.
\STATE \textbf{Output:} The recovered De Bruijn subsequence $\boldq_\boldf$.
\STATE Initialize an empty string $\mathbf{q}_{\boldf}$.
\STATE Let $U$ be the set of indices of all start locations of all $0$-runs of length $\beta$ in $\boldp_{\boldf}$.
\STATE Let $u^*=\min  U$ be the index of the first zero run in $\boldp_\boldf$..\label{line:find u*}
\FOR{$i = 0$ to $l-1$}
    
    \IF{$i \notin \bigcup_{u \in U} \{u, \dots, u+\beta-1\}$ \textbf{and} $(i - u^* + 1) \bmod{\beta} \neq 0$ and $(i - u^* ) \bmod{\beta} \neq 0$}\label{line:check if 0 run}
        \STATE \gray{\# Check if~$i$ is part of a zero run \textbf{and} check if~$i$ is an $\widetilde{\mathrm{RLL}}$ additional bit, see Lemma~\ref{lemma: there exists another RLL fixed length code}.}
            \STATE Append $\mathbf{p}_{\boldf}[i]$ to $\mathbf{q}_{\boldf}$.
    \ENDIF
\ENDFOR
\RETURN $\mathbf{q}_{\boldf}$.
\end{algorithmic}
\end{algorithm}
\begin{lemma}\label{lemma: modified de brujin sequence can be decode}
    Algorithm~\ref{algorithm: finding locations of de bruijn} correctly outputs $\mathbf{q}_{\boldf}$.
\end{lemma}
\begin{proof}
    Lemma~\ref{lemma:correctly identify modified db from pilot} guarantees that any $0$-runs of length $\beta$  in $\mathbf{p}_{\boldf}$ can be found; thus, the set of locations $U$ representing the start locations of all $0$-runs of length $\beta$ in $\mathbf{p}_{\boldf}$ is correctly identified.    
    Recall the construction process~\eqref{equation:pilot,debruijn and modified de bruijn}: 
    \begin{itemize}
        \item The $\widetilde{\mathrm{RLL}}(0,\beta-2)$ encoder inserts two $1$'s at fixed intervals into the original De Bruijn sequence $\mathbf{q}$ to generate~$\bar{\mathbf{q}}$.
        \item The sequence $\bar{\mathbf{q}}$ is transformed into $\mathbf{p}$ by inserting $0$-runs of length $\beta$ at fixed locations, see Definition~\ref{Definition: marker sequence for improved encoding}.
    \end{itemize}
    
    Since the relative distance between the additional bits inserted by $\widetilde{\mathrm{RLL}}$ encoder (the 1's in Lemma~\ref{lemma: there exists another RLL fixed length code}) and the inserted $0$-runs of length $\beta$ (Definition~\ref{Definition: marker sequence for improved encoding}) are both fixed by the construction, the position of the additional $1$'s can be calculated relative to the leftmost detected zero run at~$u^*$ (Line~\ref{line:find u*}). 
    Line~\ref{line:check if 0 run} of Algorithm~\ref{algorithm: finding locations of de bruijn} performs this calculation to filter out the additional bits and ensures the $0$-runs themselves are excluded. Consequently, the remaining bits constitute the original De Bruijn subsequence $\boldq_\boldf$.
\end{proof}
Next, we show that $\mathbf{q}_{\boldf}$ identified in Lemma~\ref{lemma: modified de brujin sequence can be decode} can be used for global alignment, thereby proving the correctness of~\textbf{B3}.
\begin{lemma}\label{lemma:global alignment for qN}
    For every $\boldf\in \cF$, $\mathbf{q}_{\boldf}$ appears only once in $\boldq$ and hence can be used to determine the global position of~$\boldf$.
\end{lemma}
\begin{proof}
    Recall from Section~\ref{section: improved encoding} that the De Bruijn sequence $\mathbf{q}$ is of length $\frac{n}{m}(1-\frac{\beta}{(1+\delta/2)\log n})\frac{\beta-2}{\beta}$. Consequently, every substring of length $\log(\frac{n(\beta-2)}{m\beta}(1-\frac{\beta}{(1+\delta/2)\log n}))$ appears in $\mathbf{q}$ exactly once.    
    Hence, to achieve global alignment, the recovered $\mathbf{q}_{\boldf}$ must exceed this unique length threshold.
    Since Algorithm~\ref{algorithm: finding locations of de bruijn} removes~$\beta$ bits every $(1+\delta/2)\log n$ bits, it follows that at most $2$ $0$-runs of length $\beta$ are removed among any $(1+\delta)\log n$ consecutive bits of pilot fragment $\boldp_\boldf$.
    Additionally, Algorithm~\ref{algorithm: finding locations of de bruijn} removes~$2$ bits every~$\beta$ bits and hence,  $\mathbf{q}_{\boldf}$ is of length at least $\frac{\beta-2}{\beta}((1+\delta)\log n-2\beta)$.
    Therefore it remains to show that 
    \begin{align}\label{eq:lemma8}
        \frac{\beta-2}{\beta}((1+\delta)\log n-2\beta) > \log\left(\frac{n(\beta-2)}{m\beta}\left(1-\frac{\beta}{(1+\delta/2)\log n}\right)\right)
    \end{align}
    for a sufficiently large~$n$.
    We expand the left-hand side (LHS) and right-hand side (RHS) of~\eqref{eq:lemma8} as follows:
    \begin{align*}
        \text{LHS} &= \left(1 - \frac{2}{\beta}\right)\left((1+\delta)\log n - 2\beta\right) \\
        &= (1+\delta)\log n - \frac{2(1+\delta)}{\beta}\log n - 2\beta + 4.
    \end{align*}
    \begin{align*}
        \text{RHS} &= \log n - \log m + \log\left(1 - \frac{2}{\beta}\right) + \log\left(1 - \frac{\beta}{(1+\delta/2)\log n}\right).
    \end{align*}
    Dividing the LHS and the RHS by $\log n$ and taking~$n$ to infinity yields the condition
    \begin{align*}
        1+\delta-\frac{2(1+\delta)}{\beta} > 1,
    \end{align*}
    which is equivalent to $\delta > \frac{2}{\beta-2}$.
    As~$\beta=o(\log(n))$ and~$\beta=\omega(1)$, $\delta > \frac{2}{\beta-2}$ holds true for any fixed $\delta > 0$ as $n\to\infty$.
\end{proof}

Finally, RLL and erasure decoding (step \textbf{C}) is executed via the following sub-steps, which are implemented in Algorithm~\ref{algorithm:retrieve Cer from aligned fragments} and proved in Lemma~\ref{lemma: rll decode} below.

\begin{enumerate}
    \item[\textbf{C1}] \textbf{Removal of $\mathrm{RLL}$-added bits:} 
    The output of step~\textbf{B3} is an estimated codeword $\hat{\boldc}_v$, which contains all fragments $\boldf\in \cF$ placed in their correct global positions, and erasures (say,~$\star$ symbols) in all other positions. 
    Since the $\mathrm{RLL}(0,\beta-1)$ encoder inserts $1$'s at prescribed locations, these positions are known from the construction. 
    We remove all such positions from $\hat{\boldc}_v$, resulting in a partially observed version of the interleaved codewords from $\cC_{\mathrm{er}}$.

    \item[\textbf{C2}] \textbf{De-interleaving:} 
    After removing the $\mathrm{RLL}$-added bits, we de-interleave the resulting sequence  to obtain estimates $\hat{\boldr}_{v(1)},\hat{\boldr}_{v(2)},\dots,\hat{\boldr}_{v(m-1)}$ of the codewords in $\cC_{\mathrm{er}}$ used to construct~$\boldc_v$ as in~\eqref{equation: interleved eq for local alignment}. 
    Each $\hat{\boldr}_{v(i)}$ contains erasures in positions that were not covered by sufficiently long fragments.

    \item[\textbf{C3}] \textbf{Erasure decoding:} 
    For each $\hat{\boldr}_{v(1)},\hat{\boldr}_{v(2)},\dots,\hat{\boldr}_{v(m-1)}$, we apply standard erasure decoding.         
\end{enumerate}
\begin{algorithm}
\caption{Extraction of $\cC_{\mathrm{er}}$-codeword estimates from globally aligned fragments}
\begin{algorithmic}[1]\label{algorithm:retrieve Cer from aligned fragments}
\STATE \textbf{Input:} An estimated codeword~$\hat{\boldc}_v$ from step \textbf{B3}.
\STATE \textbf{Output:} Estimates
$\hat{\boldr}_{v(1)},\hat{\boldr}_{v(2)},\dots,
\hat{\boldr}_{v(m-1)}$ of the codewords $\boldr_{v(1)},\boldr_{v(2)},\dots,
\boldr_{v(m-1)}\in \cC_{\mathrm{er}}$ used to create $\boldc_v$, containing erasures (i.e., $\star$ symbols).



\STATE Let $\mathcal{A}=\bigcup_{\ell=1}^{n/(m\beta)}\{m\ell\beta-(m-1),\,m\ell\beta-(m-2),\,\dots,\,m\ell\beta-1\}$.\\
\gray{\# The set~$\cA$ contains the locations of the $1$'s added to $\hat{\cC}_{\mathrm{er}}$ codewords by the
$\mathrm{RLL}(0,\beta-1)$ encoder.}
\FOR{$i=1$ to $m-1$}\label{line:algorithm3 line begin for}
    \STATE Initialize $\hat{\boldr}_{v(i)}$ as an empty string. \gray{\# An estimate of the respective $\cC_{\mathrm{er}}$-codeword.}
    \FOR{$t=0$ to $n/m-1$}
        \STATE Let $g=mt+i$.
        \IF{$g\notin \mathcal{A}$}\label{line:algorithm3 check conditions of g in mathcal A}
            \STATE Append $\hat{\boldc}_v[g]$ to $\hat{\boldr}_{v(i)}$.\label{line:algorithm3 lineappend bold cv}
        \ENDIF
    \ENDFOR
\ENDFOR\label{line:algorithm3 lineend for}
\RETURN $\hat{\boldr}_{v(1)},\hat{\boldr}_{v(2)},\dots, \hat{\boldr}_{v(m-1)}$.
\end{algorithmic}
\end{algorithm}

We comment that steps~$\textbf{C1}$ and~$\textbf{C2}$ above are implemented in Algorithm~\ref{algorithm:retrieve Cer from aligned fragments} via a linear pass over the bits of each fragment (Lines~\ref{line:algorithm3 line begin for}-\ref{line:algorithm3 lineend for}), while skipping $\textrm{RLL}(0,\beta-1)$ bits (Line~\ref{line:algorithm3 check conditions of g in mathcal A}) and sorting the remaining bits to~$\cC_\text{er}$ codewords (Line~\ref{line:algorithm3 lineappend bold cv}).
The proof of correctness of Algorithm~\ref{algorithm:retrieve Cer from aligned fragments} is as follows.

\begin{lemma}\label{lemma: rll decode}
For each $i\in[m-1]$ Algorithm~\ref{algorithm:retrieve Cer from aligned fragments} correctly outputs a partially observed version $\hat{\boldr}_{v(i)}$ of the corresponding codeword in $\cC_{\mathrm{er}}$.
\end{lemma}
\begin{proof}

Recall that the codeword $\boldc_v$~\eqref{equation: interleved eq for local alignment} is constructed by interleaving the pilot sequence and the codewords in $\bar{\cC}_{\mathrm{er}}$, i.e.,
\[
\boldc_v[mt+j] =
\begin{cases}
\boldp[t], & j=0,\\
\bar{\bolds}_{v(j)}[t], & j\in[m-1] 
\end{cases}
\]
for~$t\in\{0,1,\ldots,n/m-1\}$.
Hence, for each $i\in[m-1]$, the bits of $\bar{\bolds}_{v(i)}$ appear in positions $\{mt+i|t\in\{0,1,\ldots,n/m-1\}\}$.
Furthermore, we identify the positions of the $1$'s added by the $\mathrm{RLL}(0,\beta-1)$ encoder, as follows. 
By Lemma~\ref{lemma: there exists an RLL fixed length code} this encoder partitions each codeword of $\cC_{\mathrm{er}}$ into blocks of length $\beta-1$ and appends a~$1$ to each block, and hence 
the added $1$'s occur at positions
\[
\beta-1,\,2\beta-1,\,3\beta-1,\dots.
\]
Following the interleaving process in~\eqref{equation: interleved eq for local alignment}, these correspond to global positions
\[
\cA=\bigcup_{\ell=1}^{n/(m\beta)}\{m\ell\beta-(m-1),\,m\ell\beta-(m-2),\,\dots,\,m\ell\beta-1\},
\]
which are removed by Algorithm~\ref{algorithm:retrieve Cer from aligned fragments}. 
Consequently, for each $i\in[m-1]$, the resulting sequence $\hat{\boldr}_{v(i)}$ is obtained by deleting all RLL-added $1$'s from $\bar{\bolds}_{v(i)}$, leaving only the original symbols of the codeword in $\cC_{\mathrm{er}}$. 
Therefore, each $\hat{\boldr}_{v(i)}$ is a partially observed version of a codeword in $\cC_{\mathrm{er}}$.
\end{proof}

It remains to prove the correctness of step~\textbf{C3} in which the estimated codewords~$\hat{\boldr}_{v(1)},\ldots,\hat{\boldr}_{v(m-1)}$ are corrected to the true codewords $\boldr_{v(1)},\ldots,\boldr_{v(m-1)}\in \cC_{\text{er}}$, which requires showing that the number of erasures in each~$\hat{\boldr}_{v(i)}$ does not exceed the erasure correction capability of~$\cC_{\text{er}}$.
This is shown, with high probability as~$n\to\infty$, in Theorem~\ref{Theorem: rate calculation for improved encoding} in the following subsection, and hence the decoding succeeds with vanishing error probability.


\subsection{Rate Calculation}\label{section:calculation of rate}
In this section we complete the analysis of step \textbf{C3} by proving that the random binary linear code
$\cC_{\mathrm{er}}$ successfully corrects the erasures induced by the torn paper channel. 
We then use this guarantee to calculate the overall achievable rate.

We first emphasize that the torn paper channel is  \emph{not} memoryless. 
The torn paper channel cuts the transmitted codeword into pieces, and therefore the discarded fragments typically
induce bursts of consecutive erasures. Consequently, standard coding tools and concentration bounds for the
memoryless binary erasure channel (BEC) cannot be applied directly.

To handle this difficulty, we analyze the specific erasure patterns induced in each $\cC_{\mathrm{er}}$-codeword.
For a given codeword in $\cC_{\mathrm{er}}$, we represent the erased coordinates by a (random) subset $E \subseteq [n_{\mathrm{er}}]$, where $n_{\mathrm{er}}$ is the codeword length. 
Recall that our encoding scheme encodes the data into $m-1$ codewords of $\cC_{\mathrm{er}}$. 
For each $i \in [m-1]$, the de-interleaved sequence $\hat{\mathbf{r}}_{v(i)}$ is a partially observed codeword of $\cC_{\mathrm{er}}$, where the erasures are caused by discarded fragments.
We decode these erased symbols using the standard erasure-decoding procedure for linear codes. Specifically,
for a code defined by a parity-check matrix
$H \in \mathbb{F}_2^{r \times n_{\mathrm{er}}}$
(see Definition~\ref{definition:random linear erasure code} and
Lemma~\ref{lemma:full rank random parity matrix}), decoding reduces to solving a linear system over
$\mathbb{F}_2$ for the missing bits using Gaussian elimination. 
This procedure succeeds if and
only if the columns of $H$ indexed by the erased positions are linearly independent. 
Equivalently, for~$\cE\subseteq[n_\text{er}]$ let $H_{\mathcal{E}}$ denote the submatrix of $H$ formed by the columns in $\mathcal{E}$; decoding succeeds if and only if $\operatorname{rank}(H_{\mathcal{E}})=|\mathcal{E}|$.
The probability of satisfying this condition is analyzed in Lemma~\ref{lemma:random linear code random erasures}.

To prove that our construction succeeds despite the non-memoryless nature of the induced erasure channel, we proceed as follows.
First, we define a collection of \emph{good erasure patterns}, denoted by $\mathcal{T}$, whose sizes are strictly
below the erasure-correction capability of $\cC_{\mathrm{er}}$. We then show that the actual erasure pattern induced
by the torn paper channel belongs to $\mathcal{T}$ with probability tending to $1$ as $n \to \infty$ (see Lemma~\ref{lemma:bound on number of erasures} in Appendix~\ref{appendix for good sets}).
Then, in Lemma~\ref{lemma:random linear code random erasures} we prove that a randomly chosen parity-check matrix achieves vanishing decoding error probability for each fixed erasure pattern in~$\cT$. 

To show that the decoding error vanishes for all erasure patterns in $\cT$, let $P_e^{(i)}(H)$ denote the decoding error probability for the $i$-th $\cC_{\mathrm{er}}$-codeword when the parity-check matrix is $H$. 
Rather than analyzing one deterministic matrix directly, we evaluate 
$\mathbb{E}_H\!\left[P_e^{(i)}(H)\right],$
where the expectation is taken over uniformly random parity-check matrices $H$.
We show that this expectation
vanishes as $n \to \infty$ (see Lemma~\ref{lemma:random_matrix_capacity}). 
With the help of Markov's inequality, we then prove  that a uniformly random choice of $H$ achieves vanishing $P_e^{(i)}(H)$ with high probability (see Lemma~\ref{lemma:high_probability_existence}).

After proving that the decoding error vanishes for each individual $\cC_{\mathrm{er}}$-codeword, we extend the argument to the full torn paper code. 
Since the overall decoder fails if at least one of the $m-1$ individual $\cC_{\mathrm{er}}$-codewords fails to decode correctly, we apply a union bound over the $m-1$ codewords in $\cC_{\mathrm{er}}$.
Corollary~\ref{corollary:total_decoding_error} shows that, with probability tending to
$1$ as $n \to \infty$, the proposed encoding and decoding scheme has vanishing total error probability. 
Finally,
Theorem~\ref{Theorem: rate calculation for improved encoding} combines this vanishing decoding error with the
construction parameters to establish the achievable rate.

In what follows, for a fixed $\eta>0$ we define the \textit{set of good erasure patterns} (good set, for short), denoted by $\cT$, as the collection of all subsets of $[n_{\mathrm{er}}]$ whose size is at most~$(1-R_{\mathrm{er}}-\eta)n_{\mathrm{er}}$, i.e.,
\begin{align}\label{equation: definition of mathcal T}
     \cT = \binom{[n_\text{er}]}{\le (1-R_{\mathrm{er}}-\eta)n_{\mathrm{er}}}.
\end{align}

\begin{lemma}\label{lemma:random linear code random erasures}
Let $\cC_{\mathrm{er}}$ be the random binary linear code from Definition~\ref{definition:random linear erasure code}, with parity-check matrix $H\in\mathbb{F}_2^{r\times n_{\mathrm{er}}}$, where $r=(1-R_{\mathrm{er}})n_{\mathrm{er}}$.
For every fixed erasure pattern $\cE\in\cT$, we have
$$ \Pr_H(\rank(H_\cE) \ne |\cE|) \le 2^{-\eta n_{\mathrm{er}}}. $$
\end{lemma}

A standard proof for Lemma~\ref{lemma:random linear code random erasures} is given in Appendix~\ref{appendix}.
We proceed to compute the expected error probability induced by a random parity-check matrix.

\begin{lemma}\label{lemma:random_matrix_capacity}
    For $H \in \mathbb{F}_2^{r \times n_{\mathrm{er}}}$, where $r = (1-R_{\mathrm{er}})n_{\mathrm{er}}$, let~$P_e^{(i)}(H)$ be as explained above. 
    The expected probability of decoding error over all uniformly random parity-check matrices vanishes, that is,
    $$ \lim_{n \to \infty} \mathbb{E}_H[P_e^{(i)}(H)] = 0.$$
\end{lemma}
\begin{proof}
    Let $E_i$ be a random variable representing the set of erased coordinates in the $i$-th estimated codeword $\hat{\mathbf{r}}_{v(i)}$.
    As explained earlier, an error occurs if and only if the columns of~$H$ corresponding to the erased positions are linearly dependent, and therefore
    \begin{align*}
        P_e^{(i)}(H)=\sum_{\cE_i\in\mathcal{T}}\Pr(E_i=\cE_i)\mathbf{1}(\rank(H_{\cE_i})\ne |\cE_i|)+\sum_{\cE_i\notin\mathcal{T}}\Pr(E_i=\cE_i)\mathbf{1}(\rank(H_{\cE_i})\ne |\cE_i|),
    \end{align*}
    where~$\mathbf{1}(\cdot)$ is an indicator function.
    Since the value of~$E_i$ depends on the channel and not on the code construction of $\cC_{\mathrm{er}}$, taking the expectation $\mathbb{E}_H[\cdot]$ on both sides yields
    \begin{align*}
        \mathbb{E}_H[P_e^{(i)}(H)] &= \sum_{\cE_i\in\mathcal{T}}\Pr(E_i=\cE_i)\mathbb{E}_H[\mathbf{1}(\rank(H_{\cE_i})\ne |\cE_i|)]  + \sum_{\cE_i\notin\mathcal{T}}\Pr(E_i=\cE_i)\mathbb{E}_H[\mathbf{1}(\rank(H_{\cE_i})\ne |\cE_i|)]\\
        &=\sum_{\cE_i\in\mathcal{T}}\Pr(E_i=\cE_i)\Pr_H(\rank(H_{\cE_i})\ne |\cE_i|)  + \sum_{\cE_i\notin\mathcal{T}}\Pr(E_i=\cE_i)\Pr_H(\rank(H_{\cE_i})\ne |\cE_i|)\\
        &\overset{Lemma~\ref{lemma:random linear code random erasures}}{\le}\sum_{\cE_i\in\mathcal{T}}\Pr(E_i=\cE_i) \cdot 2^{-\eta n_{\mathrm{er}}} + \sum_{\cE_i\notin\mathcal{T}}\Pr(E_i=\cE_i) \cdot 1 \\
        &= 2^{-\eta n_{\mathrm{er}}} \sum_{\cE_i\in\mathcal{T}}\Pr(E_i=\cE_i) + \Pr(E_i \notin \mathcal{T})\le 2^{-\eta n_{\mathrm{er}}} + \Pr(E_i \notin \mathcal{T}).
    \end{align*}
    
    
    
    
    As explained in the proof of Lemma~\ref{lemma:random linear code random erasures},
    $2^{-\eta n_{\mathrm{er}}}$ goes to~$0$ exponentially fast as $n \to \infty$. 
    Furthermore, it is shown in Lemma~\ref{lemma:bound on number of erasures} in Appendix~\ref{appendix for good sets} that $\lim_{n\to\infty}\Pr(\cE_i \notin \mathcal{T}) = 0$. Therefore, the expected error probability for the $i$-th sub-codeword over the random code ensemble vanishes:
    \begin{align*}
       \lim_{n \to \infty} \mathbb{E}_H[P_e^{(i)}(H)]&= 0.\qedhere 
    \end{align*}    
\end{proof}
 The following lemma applies Markov's inequality to demonstrate that a randomly chosen parity-check matrix successfully decodes with high probability.
\begin{lemma}\label{lemma:high_probability_existence} 
    With probability which goes to~$1$ as $n\to \infty$, a randomly chosen parity-check matrix~$H$ 
    achieves a vanishing error probability for the $i$-th $\cC_\text{er}$-codeword.
\end{lemma}
\begin{proof}
    Let $\epsilon_n^{(i)} = \mathbb{E}_H[P_e^{(i)}(H)]$ denote the expected probability of decoding error; by Lemma~\ref{lemma:random_matrix_capacity} we have that $\lim_{n \to \infty} \epsilon_n^{(i)} = 0$.    
    By Markov's inequality:
    \begin{align}\label{eq:lemma12}
        \Pr_H\left(P_e^{(i)}(H) \ge \sqrt{\epsilon_n^{(i)}}\right) \le \frac{\mathbb{E}_H[P_e^{(i)}(H)]}{\sqrt{\epsilon_n^{(i)}}}=\frac{\epsilon_n^{(i)}}{\sqrt{\epsilon_n^{(i)}}} = \sqrt{\epsilon_n^{(i)}}.
    \end{align}
     
    Hence, to establish convergence in probability, fix any arbitrarily small constant $\varepsilon > 0$.
    Since $\lim_{n\to\infty} \epsilon_n^{(i)} = 0$, for every~$\varepsilon>0$ there exists $n_\varepsilon$ so that for all $n > n_\varepsilon$ we have $\sqrt{\epsilon_n^{(i)}} < \varepsilon$.     
    For any such $n$, the condition $P_e^{(i)}(H) \ge \varepsilon$ implies that $P_e^{(i)}(H) \ge \sqrt{\epsilon_n^{(i)}}$, and therefore we can bound the probability as follows:
    $$ \Pr_H\left(P_e^{(i)}(H) \ge \varepsilon\right) \le \Pr_H\left(P_e^{(i)}(H) \ge \sqrt{\epsilon_n^{(i)}}\right) \overset{\eqref{eq:lemma12}}{\le} \sqrt{\epsilon_n^{(i)}}. $$
    Taking the limit of both sides yields:
    $$ \lim_{n \to \infty} \Pr_H\left(P_e^{(i)}(H) \ge \varepsilon\right) \le \lim_{n \to \infty} \sqrt{\epsilon_n^{(i)}} = 0. $$
    Thus, a randomly chosen $H$ guarantees $P_e^{(i)}(H)\to 0$ in probability when $n\to\infty$.
\end{proof}
Having established the decoding error for each of the~$m-1$ $\cC_\text{er}$-codewords, we must now account for all $\cC_\text{er}$-codewords simultaneously. 
Since the encoding scheme encodes the data into $m-1$ codewords from $\cC_\text{er}$, an overall decoding failure occurs if even one of these codewords is decoded incorrectly.
Thus, we formalize the decoding error of torn paper channel using a union bound.
\begin{corollary}\label{corollary:total_decoding_error}
    Let~$H\in\bF_2^{r\times n_\text{er}}$ be the parity-check matrix of~$\cC_\text{er}$, and let $P_e^{\text{total}}(H)$ denote the resulting decoding error of the encoding/decoding methods in Sections~\ref{section: improved encoding} and~\ref{Section:improved decoding}.
    A parity-check matrix $H$ chosen uniformly at random achieves a vanishing total error probability as $n \to \infty$.
\end{corollary}
\begin{proof}
    By the union bound, the total decoding error is bounded by the sum of the individual decoding errors of the $m-1$ codewords:
    $$ P_e^{\text{total}}(H) \le \sum_{i=1}^{m-1} P_e^{(i)}(H),$$    
    and taking the expectation over~$H$ implies that 
    $$ \mathbb{E}_H\left[P_e^{\text{total}}(H)\right] \le \sum_{i=1}^{m-1} \mathbb{E}_H\left[P_e^{(i)}(H)\right].$$
    Recall that $\mathbb{E}_H[P_e^{(i)}(H)] \overset{n\to\infty}{\longrightarrow} 0$ for every $i \in [m-1]$ by Lemma~\ref{lemma:random_matrix_capacity}, and since $m$ is a constant, it follows that
    $$ \lim_{n \to \infty} \mathbb{E}_H\left[P_e^{\text{total}}(H)\right] = 0.$$
    Applying Markov's inequality to this total expected error, identical to the procedure in Lemma~\ref{lemma:high_probability_existence}, guarantees that a randomly chosen $H$ will successfully decode all $m-1$ codewords simultaneously with high probability.
\end{proof}
 Corollary~\ref{corollary:total_decoding_error} shows that with probability that goes to~$1$ as~$n$ goes to infinity, our encoding/decoding methods guarantee vanishing error probability.
 The following theorem establishes the resulting rate.
\begin{theorem}\label{Theorem: rate calculation for improved encoding}
    The rate of the encoding method in Section~\ref{section: improved encoding} is arbitrarily close to $(1-\frac{1}{m})(m\alpha+1)e^{-m\alpha}$ for any fixed integer $m\geq 2$.
\end{theorem}

\begin{proof}
    
    
    The total rate of the construction is given by:
    \[
        \left(1-\frac{1}{m}\right)
        \left(1-\frac{1}{\beta}\right)
        R_{\mathrm{er}},
    \]
    as shown in~\eqref{equation:rate}.
    Since $\beta=\omega(1)$, we have $1-\frac{1}{\beta}\to 1$.
    Furthermore, the parameter~$\delta$ (Definition~\ref{definition:random linear erasure code}) can be chosen arbitrarily small subject to the constraint $\delta>\frac{2}{\beta-2}$ from Lemma~\ref{lemma:global alignment for qN}.
     According to Definition~\ref{definition:random linear erasure code} and Lemma~\ref{lemma:full rank random parity matrix}, as $\eta\to 0$, $R_{\mathrm{er}}$ approaches $(m\alpha+1)e^{-m\alpha}$, and the overall achievable rate approaches
    \begin{align*}    
        \left(1-\frac{1}{m}\right)
        (m\alpha+1)e^{-m\alpha}.&\qedhere
    \end{align*}
\end{proof}
\subsection{Comparison and Optimization}\label{section:comparison}
In this section we evaluate the theoretical rate established in Theorem~\ref{Theorem: rate calculation for improved encoding} by comparing it against the torn-paper channel capacity $e^{-\alpha}$ and previous solutions from~\cite{shomorony2021torn} and~\cite{Junsheng2026torn}. 
To benchmark the achievable rate derived in Theorem~\ref{Theorem: rate calculation for improved encoding} against this limit, we performed numerical optimization over the parameter space. Specifically, we evaluated $\alpha$ across the interval $(0, 2)$ in increments of $0.01$. 
For each value of $\alpha$, the parameter $m$ was optimized over the integer domain $[2, 10^4]$ to maximize the theoretical rate. 
The results comparing our optimized rate to the channel capacity are illustrated in Fig.~\ref{Figure: comparison for rate of different methods}.
In addition, Table~\ref{tab:f1f2capacity} presents optimized values of~$m$ and the resulting rate for several representative $\alpha$ values.
\begin{table}[ht]
\centering
\begin{tabular}{c c c c c c}
\hline
$\alpha$ & Optimal $m$ in~\cite{shomorony2021torn}& Rate of~\cite{shomorony2021torn}& Optimal $m$ in this work & Rate in this work & Capacity $e^{-\alpha}$ \\
\hline
0.1 & 4 & 0.607 & 6 & 0.732 & 0.905 \\
0.2 & 3 & 0.442 & 4 & 0.607 & 0.819 \\
0.3 & 2 & 0.331 & 3 & 0.515 & 0.741 \\
0.4 & 2 & 0.262 & 3 & 0.442 & 0.670 \\
0.5 & 2 & 0.203 & 3 & 0.372 & 0.607 \\
0.6 & 2 & 0.154 & 2 & 0.331 & 0.549 \\
0.7 & 2 & 0.116 & 2 & 0.296 & 0.497 \\
0.8 & 2 & 0.086 & 2 & 0.262 & 0.449 \\
0.9 & 2 & 0.063 & 2 & 0.231 & 0.407 \\
1.0 & 2 & 0.046 & 2 & 0.203 & 0.368 \\
1.1 & 2 & 0.033 & 2 & 0.177 & 0.333 \\
1.2 & 2 & 0.024 & 2 & 0.154 & 0.301 \\
1.3 & 2 & 0.017 & 2 & 0.134 & 0.273 \\
1.4 & 2 & 0.012 & 2 & 0.116 & 0.247 \\
1.5 & 2 & 0.009 & 2 & 0.100 & 0.223
\end{tabular}
\caption{Optimal $m$ values in~\cite{shomorony2021torn} and in this work, alongside rates and capacity,\\for representative values of $\alpha \in [0.1,1.5]$.}
\label{tab:f1f2capacity}
\end{table}
\section{Improved Torn Paper Coding with Lost Pieces via Local Alignment}\label{section:compare capacity for tpclp}
Building upon the fragmentation model of the TPC, the TPC-LP framework accounts for the practical reality that some fragments may be lost during transmission.
As mentioned in Section~\ref{Section:problem definition}, the TPC-LP~\cite{ravi2024recovering} models a scenario where each fragment $\vec{X}_i$ is independently lost with probability $d(\cdot)$ (a function of fragment length). 
Consequently, the decoder receives only the unordered multiset of the surviving (error-free) fragments.

In our analysis, we specialize this model to a ``hard-threshold'' regime in the following sense---we consider a scenario where the preservation of a fragment is determined by its length relative to the total sequence length~$n$; fragments failing to meet a logarithmic length threshold of~$m\log n$ for some constant $m$ are discarded, while those exceeding this threshold are retained for decoding with probability function $1-d(\cdot)$. 
This thresholding approach is physically natural, as very short fragments are highly susceptible to being lost during transmission, while longer fragments are more likely to be retained.
Building on this framework, we establish the following definition.
\begin{definition}\label{definition:threshold deletion function}
 For an integer $m\ge 2$, a deletion probability function
$d:\mathbb{N}\to[0,1]$ is called an \emph{($m\log n$)-threshold deletion function} if
\begin{align*}
    d(\ell)&=1 \text{ if } \qquad \ell<m\log n, \text{ and }\\
    d(\ell)&\in[0,1) \text{ if } \qquad \ell\ge m\log n.
\end{align*}
Equivalently, all fragments shorter than $m\log n$ are deleted, while fragments
of length at least $m\log n$ may be retained with arbitrary length-dependent
probability $1-d(\ell)$.
\end{definition}
In what follows, we show that the scheme presented in Section~\ref{section: improved encoding} can be used in the TPC-LP, and later compute the corresponding rates that enable vanishing error probability.
In the TPC-LP model, after the torn-paper channel produces fragments of lengths
$N_1,N_2,\ldots,N_K$, each fragment is independently deleted with probability
$d(N_i)$.
By Definition~\ref{definition:threshold deletion function}, every fragment with
length smaller than $m\log n$ is deleted, while fragments of length at least $m\log n$ may be retained with arbitrary length-dependent probability $1-d(N_i)$.

The decoding procedure in Section~\ref{Section:improved decoding} only uses
fragments whose lengths are sufficiently large  for local and global alignment, while all other fragments are treated as erasures (see Algorithm~\ref{algorithm for finding locations of p in improved decoding}). 
Therefore, deleting fragments below the threshold $m\log n$ does not interfere with the decoding process, since such fragments are not used whether or not they are present.

For fragments of length at least $m\log n$, the deletion function may retain
them with any probability depending on their length; this only changes the set
of fragments available to the decoder. 
In our TPC-LP scheme, the retained fragments that are long enough are processed exactly as in Section~\ref{Section:improved decoding}:
local alignment identifies the pilot subsequence, global alignment places the fragment in its correct position in the codeword, and the remaining unavailable positions are treated as erasures in~$\cC_\text{er}$ codewords.

Thus, the coding scheme in Section~\ref{section: improved encoding} can be used without any modification for a TPC-LP channel with $(m\log n)$-threshold deletion function, after adjusting $R_\text{er}$ so that the induced erasures in~$\cC_\text{er}$ could be corrected.
In turn, this choice of $R_\text{er}$ depends on the deletion probabilities $d(\ell)$ for $\ell\ge m\log n$.
Consequently, the same encoding and decoding scheme is valid for any $(m\log n)$-threshold
deletion function: all deleted fragments are treated as
erasures, while all retained (sufficiently long) fragments are aligned and used
for decoding.
The capacity of the TPC-LP is as follows.

\begin{lemma}\label{lemma:capacity of tplcp}\cite[Corollary 2]{ravi2024recovering}
Assuming all limits exist and are finite, the capacity of the TPC-LP is 
\begin{equation} \label{eq:tpc_capacity}
    C = \alpha \int_{1}^{\infty} (\kappa - 1) \left( 1 - \hat{d}(\kappa) \right) h(\kappa) \, d\kappa, \tag{13}
\end{equation}
where $\hat{d}(\kappa) \triangleq \lim_{n \to \infty} d(\kappa \log n)$ and $h(\kappa) \triangleq \lim_{n \to \infty} \Pr(N_i = \kappa \log n) \log n$.
\end{lemma}
\begin{theorem}\label{theorem:arbitrarily small gap generalized tpclp}
For~$m\ge 2$, let $d(\cdot)$ be an $(m\log n)$-threshold deletion function (Definition~\ref{definition:threshold deletion function}).
Assume that the normalized deletion function
\[
    \hat d(\kappa)= \lim_{n\to\infty} d(\kappa\log n)
\]
exists for all $\kappa\ge m$.
Then, for sufficiently large $m$, the encoding scheme in
Section~\ref{section: improved encoding} achieves an arbitrarily small
additive gap to the TPC-LP capacity~\eqref{eq:tpc_capacity}.
\end{theorem}

\begin{proof}
By Corollary~\ref{lemma:capacity of tplcp}, the TPC-LP capacity is
\[
    C
    =
    \alpha
    \int_{1}^{\infty}
    (\kappa-1)(1-\hat d(\kappa))h(\kappa)\,d\kappa.
\]
Since $d(\cdot)$ is a $(m\log n)$-threshold deletion function, all fragments shorter than
$m\log n$ are deleted, hence $\hat d(\kappa)=1$ for $\kappa<m$, and therefore,
\[
    C
    =
    \alpha
    \int_{m}^{\infty}
    (\kappa-1)(1-\hat d(\kappa))h(\kappa)\,d\kappa.
\]
For a torn paper channel with break probability~$p_n$ such that~$\lim_{n\to\infty}p_n\log n= \alpha$, the fragment lengths~$N_i$ are i.i.d. geometric random variables with parameter~$p_n$, and hence $\mathbb{E}[N_i] \triangleq \ell_n = 1/p_n$. Using the definition of~$h(\kappa)$ from Lemma~\ref{lemma:capacity of tplcp}, we calculate the asymptotic density (\cite[Eq.(14)]{ravi2024recovering}):
\begin{align}\label{eq:hkappa}
    h(\kappa) &= \lim_{n\to\infty} \Pr(N_i = \kappa\log n)\log n = \alpha e^{-\alpha\kappa}.
\end{align}
Thus,
\[
    C
    =
    \alpha^2
    \int_{m}^{\infty}
    (\kappa-1)(1-\hat d(\kappa))e^{-\alpha\kappa}\,d\kappa.
\]
The decoder uses the retained fragments of length at least $m\log n$. 
As in~\cite[Definition~1]{ravi2024recovering}, the asymptotic expected fraction of codeword coordinates covered by such retained fragments is denoted by
\begin{align}\label{equation:f_d(mlogn)}
        F_d(m\log n)
    \triangleq
    \lim_{n\to\infty}
    \mathbb{E}\left[
        \frac{1}{n}
        \sum_{i=1}^{K}
        N_i
        \mathbf{1}(\vec{X}_i\in Y,\;N_i\ge m\log n)
    \right].
\end{align}
where $\vec{X}_i\in Y$ means that the fragment is retained by the TPC-LP channel.

By the definition of the indicator function~$\mathbf{1}(\cdot)$, we have that $\mathbf{1}(\vec{X}_i\in Y,\;N_i\ge m\log n)=
\mathbf{1}(\vec{X}_i\in Y)\mathbf{1}(N_i\ge m\log n)$. 
Furthermore, by Lemma~\ref{lemma:number of fragments} in Appendix~\ref{appendix for good sets}, the total number of fragments $K$ tightly concentrates around $n/\ell_n$, where $\ell_n = \mathbb{E}[N_1]=1/p_n$. 
Since the fragment lengths $N_1, N_2, \dots, N_K$ are identically distributed, the expected value inside the sum in~\eqref{equation:f_d(mlogn)} is identical for all $i$. 
Thus, we can drop the index $i$ and evaluate the expectation using the first fragment $N_1$ as a representative, multiplied by~$n/\ell_n$.
Namely, we replace the sum over $K$ with $n/\ell_n$ copies of the expected value for~$N_1$; the factor of~$n$ cancels out, and we obtain 
\begin{align}\label{eq:fdmlogn}
F_d(m\log n)
&=
\lim_{n\to\infty}
\frac{1}{\ell_n}
\mathbb{E}\left[
N_1
\mathbf{1}({N_1\ge m\log n})
\mathbf{1}{(\vec{X}_1\in Y)}
\right] \nonumber\\
&=
\lim_{n\to\infty}
\frac{1}{\ell_n}
\sum_{a=1}^{\infty}
a\,
\Pr\left(
N_1
\mathbf{1}{(N_1\ge m\log n)}
\mathbf{1}{(\vec{X}_1\in Y)}
=a
\right).
\end{align}

For any \(a\ge 1\), the event
\[
N_1
\mathbf{1}{(N_1\ge m\log n)}
\mathbf{1}{(\vec{X}_1\in Y)}
=a
\]
occurs if and only if
\[
N_1=a\qquad\text{and}\qquad a\ge m\log n\qquad\text{and
}\qquad \vec{X}_1\in Y.
\]
 By the definition of the TPC-LP channel, conditioned on the event
\(N_1=a\), the fragment \(\vec X_1\) is retained with probability
\(1-d(a)\). Therefore,
\[
\Pr(\vec X_1\in Y\mid N_1=a)
=
1-d(a),
\]
and hence, by the definition of conditional probability,
\[
\Pr(N_1=a,\vec X_1\in Y)
=
\Pr(N_1=a)\Pr(\vec X_1\in Y\mid N_1=a)
=
\Pr(N_1=a)(1-d(a)).
\]
Substituting this into \eqref{eq:fdmlogn} gives
\begin{align}\label{eq:fdmlogn2}
F_d(m\log n)
&=
\lim_{n\to\infty}
\frac{1}{\ell_n}
\sum_{a=m\log n}^{\infty}
a\Pr\left(N_1=a,\vec{X}_1\in Y\right) \nonumber\\
&=
\lim_{n\to\infty}
\frac{1}{\ell_n}
\sum_{a= m\log n}^{\infty}
a(1-d(a))\Pr(N_1=a).
\end{align}

Now, let 
\[
    \Delta_n=\frac{1}{\log n}
    \qquad\text{ and } \qquad
\mathcal K_n=\left\{m+i\cdot\Delta_n:i\in\mathbb N\cup\{0\}\right\},
\]
and hence,
\begin{align}\label{eq:riemannsum}
    \eqref{eq:fdmlogn2}
    &=
    \lim_{n\to\infty}
    \frac{\log n}{\ell_n}
    \sum_{a= m\log n}^{\infty}
    \frac{a}{\log n}
    \left(1-d\left(\frac{a}{\log n}\cdot\log n\right)\right)
    \Pr\left(N_1=\frac{a}{\log n}\cdot\log n\right)
    \log n\cdot \Delta_n \nonumber\\
    &=\lim_{n\to\infty}
    \frac{\log n}{\ell_n}\sum_{\kappa\in\mathcal K_n}\kappa
    \left(1-d(\kappa\log n)\right)
    \Pr(N_1=\kappa\log n)\log n\cdot
    \Delta_n.
\end{align}
To proceed with the computation of~\eqref{eq:riemannsum}, we define a sequence of piece-wise constant functions $f_n: [m, \infty) \to \mathbb{R}$ as follows.
For each  $\kappa \in \mathcal{K}_n$, we set the value of $f_n(x)$ over the interval $[ \kappa, \kappa + \Delta_n )$ to match the corresponding summand:
$$f_n(x) = \kappa \left(1-d(\kappa\log n)\right) \Pr(N_1=\kappa\log n)\log n \text{ if } x \in [\kappa, \kappa + \Delta_n) \text{ for some }\kappa\in\cK_n;$$
since the intervals $\{[\kappa,\kappa+\Delta_n)\vert \kappa\in\cK_n\}$ are disjoint and cover~$[m,\infty)$ for every~$n$, the function~$f_n(x)$ is well defined for every~$x\in[m,\infty)$ and every~$n$.
Since $f_n(x)$ 
is fixed on each such interval,
its integral over $[m, \infty)$ exactly evaluates to the Riemann  sum:$$\int_m^\infty f_n(x) \, dx = \sum_{\kappa\in\mathcal K_n}\kappa \left(1-d(\kappa\log n)\right) \Pr(N_1=\kappa\log n)\log n \cdot \Delta_n,$$
and hence we ought to evaluate:
\begin{align}\label{equation:getAlphaOut}
\eqref{eq:riemannsum}=\lim_{n\to\infty} \frac{\log n}{\ell_n} \int_m^\infty f_n(x)  dx=\alpha\lim_{n\to\infty} \int_m^\infty f_n(x)  dx,
\end{align}
where the last equality follows since $\ell_n=\mathbb{E}[N_1]=1/p_n$ and $p_n\log n\to\alpha$ imply that $\frac{\log n}{\ell_n}=p_n\log n\xrightarrow{n\to\infty}\alpha$, and by the product law for limits.

To switch between the integral and the limit in~\eqref{equation:getAlphaOut}, we utilize Lebesgue's dominated convergence theorem\footnote{Lebesgue's dominated convergence theorem~\cite{wiki:Dominated_convergence_theorem} states that if $f_n \to f$ pointwise and $|f_n|\le g$ for some integrable function $g$, then $\lim_{n\to\infty}\int f_n
=
\int \lim_{n\to\infty} f_n
=
\int f$.
}. 
The first step in applying this theorem is to prove pointwise convergence of the sequence $f_n(x)$, i.e. $$\lim_{n\to\infty}f_n(x)=f(x)$$
for some function $f(x)$.
Indeed, as shown in~\eqref{eq:tpc_capacity}, we have $\hat{d}(\kappa) \triangleq \lim_{n \to \infty} d(\kappa \log n)$ and 
\(h(\kappa)\triangleq \lim_{n \to \infty}\Pr(N_1=\kappa\log n)\log n\), thus for every $x\in [m,\infty)$, $$\lim_{n\to\infty}f_n(x) = x (1 - \hat{d}(x)) h(x)\triangleq f(x).$$ 
Next, we ought to prove the existence of an integrable dominating function $g(x)$ such that $|f_n(x)|\le g(x)$ for all sufficiently large~$n$ and all $x\in[m,\infty)$.
By setting $a=x \log n$ in Lemma~\ref{lemma:geometric-envelope}, given in Appendix~\ref{section:lemma useful for thm2}, we have $$\Pr(N_1=x\log n)\log n\le Ce^{-cx}$$ for some constant $C,c>0$.
Thus, $$f_n(x)\le x(1-d(x\log n))Ce^{-cx}\le Cxe^{-cx}\triangleq g(x).$$
The dominating function $g(x) = C x e^{-c x}$ is strictly integrable over our domain:$$\int_m^\infty g(x) \, dx = C \int_m^\infty x e^{-c x} \, dx < \infty.$$
Since $g \in L^1([m, \infty))$, the dominated convergence theorem guarantees that the limit and integral can be switched: $$\lim_{n\to\infty} \int_m^\infty f_n(x) \, dx = \int_m^\infty \lim_{n\to\infty}f_n(x) \, dx= \int_m^\infty f(x) \, dx = \int_m^\infty x (1 - \hat{d}(x)) h(x) \, dx.$$

Hence,
\begin{align}\label{equation:R_er}
        F_d(m\log n)
    =
    \alpha
    \int_m^\infty
    \kappa(1-\hat d(\kappa))h(\kappa)\,d\kappa.
\end{align}

By the same erasure-reduction argument as in
Theorem~\ref{Theorem: rate calculation for improved encoding}, the retained fragments induce an effective erasure channel on the codewords of
$\cC_{\mathrm{er}}$. 
The asymptotic fraction of un-erased coordinates is
precisely $F_d(m\log n)$, and therefore the erasure code
$\cC_{\mathrm{er}}$ may be chosen with rate arbitrarily close to
$F_d(m\log n)$ while still guaranteeing vanishing decoding error
probability.
Furthermore, only $m-1$ of the $m$ interleaved subsequences
carry information symbols, since the remaining one is
the pilot sequence. 
Hence, the overall rate is reduced by a multiplicative
factor of $(1-1/m)$. 
Remark~\ref{remark:tpclp_erasure_patterns} from
Appendix~\ref{appendix for good sets} shows that the TPC-LP deletion function
only changes the fraction of retained fragments and does not change the way
erasures are distributed across the interleaved codewords of $\mathcal{C}_{\mathrm{er}}$. 
Therefore the same local-alignment, global-alignment, and erasure-decoding analysis in Section~\ref{Section:improved decoding} continues to apply.
Hence, the same argument used in Theorem~\ref{Theorem: rate calculation for improved encoding}
applies here as well.
Consequently, the achievable rate is
\[
    R_{\mathrm{ach}}
    =
    \left(1-\frac{1}{m}\right)
    \alpha
    \int_{m}^{\infty}
    \kappa(1-\hat d(\kappa))h(\kappa)\,d\kappa.
\]
To compare this rate to the TPC-LP capacity observe the following. 
\begin{align*}
    C-R_{\mathrm{ach}}
    &=
    \alpha\int_m^\infty
    (\kappa-1)(1-\hat d(\kappa))h(\kappa)\,d\kappa \\
    &\quad -
    \left(1-\frac{1}{m}\right)
    \alpha\int_m^\infty
    \kappa(1-\hat d(\kappa))h(\kappa)\,d\kappa \\
    &=
    \alpha\int_m^\infty
    \left(\kappa-1-\kappa+\frac{\kappa}{m}\right)
    (1-\hat d(\kappa))h(\kappa)\,d\kappa \\
    &=
    \frac{\alpha}{m}
    \int_m^\infty
    (\kappa-m)(1-\hat d(\kappa))h(\kappa)\,d\kappa.
\end{align*}
Since $0\le 1-\hat d(\kappa)\le 1$ and $h(\kappa)=\alpha e^{-\alpha\kappa}$~\eqref{eq:hkappa}, it follows that
\[
    C-R_{\mathrm{ach}}
    \le
    \frac{\alpha}{m}
    \int_m^\infty
    (\kappa-m)\alpha e^{-\alpha\kappa}\,d\kappa.
\]
Evaluating the integral yields
\[
    \frac{\alpha^2}{m}
    \int_m^\infty
    (\kappa-m)e^{-\alpha\kappa}\,d\kappa
    =
    \frac{e^{-\alpha m}}{m},
\]
therefore
\[
    0\le C-R_{\mathrm{ach}}
    \le
    \frac{e^{-\alpha m}}{m},
\]
and thus the difference~$C-R_\text{ach}$ vanishes as $m\to\infty$. 
Hence, for sufficiently large
$m$, the encoding scheme in Section~\ref{section: improved encoding} achieves
an arbitrarily small additive gap to the TPC-LP capacity.
\end{proof}

\begin{example}
    Consider a TPC-LP with \(\alpha=0.1\) and\footnote{
    Notice that unlike the TPC channel considered in Sections~\ref{section: improved encoding},
    where \(m\) served as a parameter that can be optimized, in the TPC-LP setting~$m$ is a parameter defined by the channel, which determines the deletion
    probability via Definition~\ref{definition:threshold deletion function}.} $m=10$, and the following
    length-dependent deletion rule:
    \[
    d(\ell)=
    \begin{cases}
    1, & \ell<10\log n,\\
    0.2, & 10\log n\le \ell<20\log n,\\
    0, & \ell\ge 20\log n.
    \end{cases}
    \]
    Equivalently, the asymptotic retention probability is
    \[
    1-\hat d(\kappa)=
    \begin{cases}
    0, & \kappa<10,\\
    0.8, & 10\le \kappa<20,\\
    1, & \kappa\ge 20.
    \end{cases}
    \]
    Since the asymptotic fragment length density is
    \[
    h(\kappa)=\alpha e^{-\alpha\kappa}=0.1e^{-0.1\kappa},
    \]
    the TPC-LP capacity from Lemma~\ref{lemma:capacity of tplcp} is
    \[
    C
    =
    0.1^2 \left(
    0.8 \int_{10}^{20}(\kappa-1)e^{-0.1\kappa}\,d\kappa
    +
    \int_{20}^{\infty}(\kappa-1)e^{-0.1\kappa}\,d\kappa
    \right).
    \]
    which evaluates numerically to
    \begin{align*}
    C\approx 0.6377.
    \end{align*}
    We employ a random linear binary code as the erasure code, whose rate is chosen according to~\eqref{equation:R_er} as follows.     
\[
R_{\mathrm{er}}
\overset{\eqref{equation:R_er}}{=}
0.1^2 \left(
0.8 \int_{10}^{20} \kappa e^{-0.1\kappa}\,d\kappa
+
1 \int_{20}^\infty \kappa e^{-0.1\kappa}\,d\kappa
\right),
\]
which numerically evaluates to:
\begin{align*}
R_{\mathrm{er}} \approx 0.6698.
\end{align*}
The overall achievable rate is then calculated as:
\begin{align*}
R_{\mathrm{ach}}
&=
\left(1-\frac{1}{10}\right)R_{\mathrm{er}} \approx 0.6028.
\end{align*}
 Comparing the theoretical capacity $C$ against our achieved rate $R_{\mathrm{ach}}$, we observe the gap:
\[
C - R_{\mathrm{ach}} \approx 0.0348.
\]
This example illustrates that the proposed construction can operate
    close to the TPC-LP capacity under a small additive gap.
\end{example}
\printbibliography
\appendices
\section{Omitted proof}\label{appendix}
\begin{proof}[Proof of Lemma~\ref{lemma:random linear code random erasures}]
Fix~$\cE\in\cT$, let~$t=|\cE|$,
and let $H_\cE$ denote the submatrix of $H$ formed by the columns indexed by $\cE$. 
A binary linear code successfully decodes the erasure pattern $\cE$ if and only if the columns of $H_\cE$ are linearly independent, which is equivalent to $\rank(H_\cE) = t$. 
Since the entries of $H$ are chosen independently and uniformly from $\mathbb{F}_2$, we can analyze the linear independence of $H_\cE$ by building the submatrix column by column.
Conditioned on the first $i$ columns being linearly independent, their span contains exactly $2^i$ vectors out of the $2^r$ vectors in $\mathbb{F}_2^r$. 
The probability that the $(i+1)$-th randomly chosen column belongs to this span (an individual failure) is $2^i / 2^r = 2^{i-r}$. 

Since rank deficiency occurs if any individual column belongs to the span of its predecessors, we can apply the union bound to sum these individual failure probabilities. 
Thus, the probability that $H_\cE$ is rank deficient is at most:
$$ \Pr_H(\rank(H_\cE) \ne t) \le \sum_{i=0}^{t-1}2^{i-r} < 2^{t-r}. $$
Since $r=(1-R_{\mathrm{er}})n_{\mathrm{er}}$ and $t\le(1-R_{\mathrm{er}}-\eta)n_{\mathrm{er}}$, it follows that $t-r \le -\eta n_{\mathrm{er}}$, and hence we obtain the bound $\Pr_H(\rank(H_\cE) \ne |\cE|) \le 2^{-\eta n_{\mathrm{er}}}$. 
Since $n_{\mathrm{er}} = \frac{(\beta-1)n}{m\beta}$ where~$\beta=o(\log(n))$, $\beta=\omega(1)$, and~$m=O(1)$, we have that $n_{\mathrm{er}}\to\infty$ whenever $n\to\infty$. 
Therefore, since $\eta > 0$, this failure probability tends to zero exponentially fast when $n\to\infty$.
\end{proof}
\section{Proofs regarding good erasure patterns}\label{appendix for good sets}
In this appendix, we show that the erasure pattern induced
by the torn paper channel belongs to $\mathcal{T}$ (see~\eqref{equation: definition of mathcal T}) with probability tending to $1$ as $n \to \infty$.
The following Definition~\ref{Definition: coverage of fragment}, Lemma~\ref{lemma:number of fragments} and Lemma~\ref{lemma:concentration of coverage} are used to prove Lemma~\ref{lemma:bound on number of erasures}.
\begin{definition}\label{Definition: coverage of fragment}\cite[Definition~1]{shomorony2021torn}
    For a constant $\gamma$, the \textit{coverage} $V_\gamma$ is the fraction of output bits which reside in fragments of length at least $\gamma\log n$, i.e.,
    \begin{equation*}
        V_\gamma=\frac{1}{n}\sum_{i=1}^{K}N_i\mathbf{1}({N_i\geq \gamma\log n }),
    \end{equation*}
    where $\mathbf{1}({N_i\geq \gamma\log n })$ is the indicator function of the event $N_i\geq \gamma\log n $.
\end{definition}
\begin{lemma}\label{lemma:number of fragments}
 Let $K$ be the smallest index such that
\[
    \sum_{i=1}^{K} N_i \ge n,
\]
where $N_1,N_2,\ldots$ are i.i.d. $\operatorname{Geometric}(p_n)$ random variables representing the lengths of broken fragments. Then
\[
    K = O\!\left(\frac{n}{\log n}\right)
\]
with high probability.
\end{lemma}

\begin{proof}
For each $t\in\{1,\dots,n-1\}$, let $B_t$ be the indicator of the event that a break occurred between $x_t$ and $x_{t+1}$. Then
\[
    B_t \sim \operatorname{Ber}(p_n),
\]
and the random variables $B_1,\dots,B_{n-1}$ are independent. Since every break
increases the number of fragments by one, we have
\[
    K = 1+\sum_{t=1}^{n-1} B_t.
\]
Since $\lim_{n\to\infty} p_n\log n=\alpha$, there exists a constant $c>0$ such that,
for all sufficiently large $n$,
\[
    p_n \le \frac{c}{\log n},
\]
and hence
\[
    \mathbb{E}\left[\sum_{t=1}^{n-1} B_t\right]
    =
    (n-1)p_n
    \le
    \frac{c(n-1)}{\log n}.
\]
By Chernoff's bound, for all sufficiently large $n$,
\[
    \Pr\left(
        \sum_{t=1}^{n-1}B_t
        >
        2\frac{c(n-1)}{\log n}
    \right)
    \le
    \exp\left(-\frac{c(n-1)}{3\log n}\right),
\]
which tends to zero as $n\to\infty$. Therefore, with high probability,
\begin{align*}
    K
    &=
    1+\sum_{t=1}^{n-1}B_t
    \le
    1+2\frac{c(n-1)}{\log n}
    =
    O\!\left(\frac{n}{\log n}\right).\qedhere
\end{align*}
\end{proof}
\begin{lemma}\label{lemma:concentration of coverage}\cite[Lemma~6]{shomorony2021torn}
For a constant~$\gamma$, the coverage~$V_\gamma$ satisfies the following two properties
    \begin{align}\label{Equation: limit of coverage}
        \Pr\!\left( \left| V_{\gamma} - (\alpha\gamma + 1)e^{-\alpha\gamma} \right| > \epsilon \right)
\;\xrightarrow[n\to\infty]{} 0,\qquad \text{and} \qquad \lim_{n \to \infty} \mathbb{E}[V_\gamma]
    = (\alpha \gamma + 1)e^{-\alpha \gamma}.
\end{align}
\end{lemma}

\begin{lemma}\label{lemma:bound on number of erasures}
Consider the decoding procedure after global alignment (Step \textbf{B3}) and de-interleaving (Step \textbf{C2}).
For each $i\in[m-1]$, let $E_i$ be a random variable representing the set of erased coordinates in the $i$-th estimated codeword $\hat{\mathbf{r}}_{v(i)}$.
Let $R_{\mathrm{er}}$ and $n_{\mathrm{er}}$ be the rate and block length of the random binary linear code $\cC_{\mathrm{er}}$ in Definition~\ref{definition:random linear erasure code}. 
Then, for any fixed $\eta>0$ and for every $i \in [m-1]$, we have that~$\Pr(E_i\in\cT)\overset{n\to\infty}{\longrightarrow }1$.
\end{lemma}

\begin{proof}
Let $G\subseteq[n]$ be the set of coordinates of the interleaved sequence $\boldc_v$ that lie in fragments of length strictly smaller than $\gamma\log n$. 
These fragments are discarded by the decoder in Step~$\mathbf{A}$, and by Definition~\ref{Definition: coverage of fragment} we have that
\[
    |G|
    =
    \sum_{j=1}^{K}N_j\mathbf{1}({N_j<\gamma\log n})
    =
    n(1-V_\gamma).
\]
After global alignment, the decoder obtains an estimate $\hat{\boldc}_v$ that agrees with $\boldc_v$ on all coordinates not in $G$, and contains erasures on the coordinates in $G$. 
We now analyze how these erasures appear in each $\cC_{\mathrm{er}}$-codeword after removing the $\mathrm{RLL}(0,\beta-1)$ added $1$'s and de-interleaving.
For each $i\in[m-1]$, define
\[
    \cD_i
    =
    \{mt+i:0\le t\le n/m-1,\ t\not\equiv (\beta-1) \bmod{\beta}\}.
\]
The set $\cD_i$ consists exactly of the global positions in $\boldc_v$ corresponding to the symbols of the $i$-th codeword in $\cC_{\mathrm{er}}$.
The number of erasures in $\hat{\boldr}_{v(i)}$ is precisely $|E_i| = |G\cap \cD_i|$.

For any interval $\cI\subseteq[n]$, there exists a constant $c_1=c_1(m)$ such that 
\begin{equation}\label{equation: I cap D_i}
        |\cI\cap \cD_i|
    \le
    \frac{\beta-1}{m\beta}|\cI|+c_1\beta.
\end{equation}
The correctness of~\eqref{equation: I cap D_i} is proved by observing that \(\cD_i\) contains exactly \(\beta-1\) positions among every
\(m\beta\) consecutive positions. Hence the first term reflects the density of
\(\cD_i\), whereas the additive term \(c_1\beta\) accounts for the incomplete
\(m\beta\)-blocks intersecting the two endpoints of \(\cI\).
Therefore, since $G$ is a union of at most $K$ fragments (each of which is an interval), we obtain
\[
    |E_i| = |G\cap \cD_i|=\sum_{\boldf\in G}|\operatorname{entries}(\boldf)\cap \cD_i| \overset{\eqref{equation: I cap D_i}}{\le}\sum_{\boldf\in G} \left(\frac{\beta-1}{m\beta}|\operatorname{entries}(\boldf)|+c_1\beta\right)
    =
    \frac{\beta-1}{m\beta}|G|+c_1\beta K,
\]
where $\operatorname{entries}(\boldf)$ is the segment (i.e., consecutive subsets) of~$[n]$ corresponding to the fragment~$\boldf$.
By Lemma~\ref{lemma:number of fragments}, for all sufficiently large $n$, with high probability, $K \le c_2\frac{n}{\log n}$ for some constant $c_2>0$. Since $\beta=o(\log n)$, we have $c_1\beta K \le c_1c_2\frac{\beta n}{\log n} = o(n)$.
Therefore, with high probability we have that for every $i\in[m-1]$,
\[
    |E_i|
    \le
    \frac{\beta-1}{m\beta}n(1-V_\gamma)+o(n).
\]

Recall that the block length of $\cC_{\mathrm{er}}$ is $n_{\mathrm{er}} = |\cD_i| = \frac{(\beta-1)n}{m\beta}$.
Thus, dividing by $n_{\mathrm{er}}$, the fraction of bits erased in every $\hat{\boldr}_{v(i)}$ is at most
\[
    \frac{|E_i|}{n_{\mathrm{er}}}
    \le
    1-V_\gamma+o(1),
\]
where the~$o(1)$ expression follows since~$\beta=\omega(1)$, and since~$m$ is a constant.
Now, by the left-hand side of~\eqref{Equation: limit of coverage} in Lemma~\ref{lemma:concentration of coverage}, $V_\gamma$ converges in probability towards $(\alpha\gamma+1)e^{-\alpha\gamma}$. Hence, for any fixed $\eta>0$, we have $V_\gamma \ge (\alpha\gamma+1)e^{-\alpha\gamma}-\eta$ with high probability as $n \to\infty$. 
By Lemma~\ref{lemma:full rank random parity matrix}, the rate of $\cC_{\mathrm{er}}$ is $R_{\mathrm{er}} = (\alpha\gamma+1)e^{-\alpha\gamma}-3\eta$. 
, and hence, substituting $R_{\mathrm{er}}$ yields:
\[
    V_\gamma \ge R_{\mathrm{er}} + 2\eta.
\]
Consequently, the fraction of erased bits is bounded by:
\[
    \frac{|E_i|}{n_{\mathrm{er}}} \le 1 - (R_{\mathrm{er}} + 2\eta) + o(1) = 1 - R_{\mathrm{er}} - 2\eta + o(1).
\]
For all sufficiently large $n$, the $o(1)$ term is strictly less than $\eta$. Therefore, the size satisfies $|E_i| \le (1-R_{\mathrm{er}}-\eta)n_{\mathrm{er}}$ with high probability for every $i \in [m-1]$. This implies that the random variable $E_i$ falls into the good set $\cT$ with high probability, yielding $\Pr(E_i \in \cT) \to 1$.
\end{proof}

\begin{remark}\label{remark:tpclp_erasure_patterns}
The argument in Lemma~\ref{lemma:bound on number of erasures} extends to the
TPC-LP setting since the additional deletion step does not change the
structure of the induced erasure patterns. 
In the coding scheme presented in Section~\ref{Section:improved decoding}, the decoder discards all fragments whose lengths are below the alignment
threshold, and these discarded fragments induce erasures that are unions of
whole intervals in the transmitted codeword. 
In the TPC-LP setting considered in Section~\ref{section:compare capacity for tpclp}, the deletion function also removes all fragments whose lengths are below
the threshold $m\log n$, exactly as the ordinary TPC decoder would discard
unusable short fragments. 
For longer fragments, the deletion function may remove additional fragments,
but this deletion decision depends only on the fragment length and not on where
the fragment originated in the codeword. Hence, every $\mathcal{C}_{\mathrm{er}}$-codeword is affected by the erasures evenly.
Therefore, all deletions still correspond to
removing whole intervals in the original transmitted word. Since each such
interval intersects the sets $D_i$ according to the same counting bound as
in~\eqref{equation: I cap D_i}, the resulting erasures remain evenly
distributed across the $m-1$ interleaved codewords  of
$\mathcal{C}_{\mathrm{er}}$, up to the same negligible boundary term. Thus, the
TPC-LP deletion function only changes the asymptotic fraction of coordinates
covered by usable retained fragments. Consequently, the same random
parity-check matrix argument from Lemmas~\ref{lemma:random linear code random erasures}--\ref{lemma:high_probability_existence} still provides vanishing decoding error probability,
provided that the rate $R_{\mathrm{er}}$ is chosen according to the $(m\log n)$-threshold deletion function and coverage $F_d(m\log n)$. 
\end{remark}

\section{lemma for theorem~\ref{theorem:arbitrarily small gap generalized tpclp}}\label{section:lemma useful for thm2}
\begin{lemma}\label{lemma:geometric-envelope}
Let \(N_1\) be a geometric random variable with parameter \(p_n\), that is,
\[
    \Pr(N_1=a)=p_n(1-p_n)^{a-1},\qquad a=1,2,\ldots .
\]
Assume that
\[
    p_n\log n\to \alpha
\]
for some \(\alpha\in(0,\infty)\). Then there exist constants \(C,c>0\) and \(n_0\in\mathbb N\) such that, for all \(n\ge n_0\) and all integers \(a\ge 1\),
\[
    \log n\cdot \Pr(N_1=a)
    \le
    Ce^{-c\frac{a}{\log n}}.
\]
\end{lemma}

\begin{proof}
Since \(p_n\log n\to \alpha\in(0,\infty)\), there exist constants
\(A,B>0\) and \(n_0\in\mathbb N\) such that, for all \(n\ge n_0\),
\[
    A \le p_n\log n \le B.
\]
In particular,
\[
    p_n\le \frac{B}{\log n}
    \qquad\text{and}\qquad
    p_n\ge \frac{A}{\log n}.
\]

For \(a\ge 1\), using \(1-x\le e^{-x}\), we have
\[
\begin{aligned}
    \log n\cdot \Pr(N_1=a)
    &= \log n\cdot p_n(1-p_n)^{a-1} \\
    &\le \log n\cdot  p_n e^{-p_n(a-1)} \\
    &\le B e^{-p_n(a-1)}.
\end{aligned}
\]
Using \(p_n\ge A/\log n\), we get
\[
    e^{-p_n(a-1)}
    \le
    e^{-\frac{A(a-1)}{\log n}}.
\]
Therefore
\[
    \log n\cdot \Pr(N_1=a)
    \le
    Be^{-\frac{A(a-1)}{\log n}}.
\]
Now
\[
    -\frac{A(a-1)}{\log n}
    =
    -\frac{Aa}{\log n}+\frac{A}{\log n}.
\]
For all sufficiently large \(n\), \(\log n\ge 1\), so
\[
    e^{\frac{A}{\log n}}\le e^A.
\]
Hence
\[
    \log n\cdot \Pr(N_1=a)
    \le
    Be^A
    e^{-A\frac{a}{\log n}}.
\]
Thus the desired bound holds with $C=Be^A$ and $c=A$.
\end{proof}
\end{document}